\documentclass[11pt]{article}
\usepackage[margin=1in]{geometry}
\usepackage{csquotes}
\usepackage{amsfonts,amsmath,amssymb,dsfont,mathtools}
\usepackage[hypertexnames=false]{hyperref}
\usepackage{amsthm,comment}
\usepackage[inline, shortlabels]{enumitem}
\usepackage{xcolor}
\usepackage{tikz}
\usepackage{appendix}

\allowdisplaybreaks

\hypersetup{ 
    colorlinks=true,
    citecolor=black,
    linkcolor=black,
    urlcolor=black
}

\usepackage[noabbrev,capitalise,nameinlink]{cleveref}
\crefname{enumi}{}{}
\crefname{equation}{}{}
\crefname{claim}{Claim}{Claims}

\newcommand{\Bin}{\mathrm{Bin}}

\newtheorem{theorem}{Theorem}[section]
\newtheorem{lemma}[theorem]{Lemma}

\newtheorem{fact}[theorem]{Fact}
\newtheorem{remark}[theorem]{Remark}
\newtheorem{claim}[theorem]{Claim}
\newtheorem{definition}[theorem]{Definition}

\newcommand{\Prob}[1]{\mathbf{Pr}\left[#1\right]}

\newcommand{\Expc}[1]{\mathbf{E}\left[#1\right]}

\newcommand{\inote}[1]{\textcolor{blue}{(\textbf{Isa:} #1)}}

\newcommand{\anote}[1]{\textcolor{purple}{(\textbf{Andy:} #1)}}

\renewcommand{\leq}{\leqslant}

\renewcommand{\geq}{\geqslant}

\newcommand{\gossip}{{\sc{gossip}}}
\newcommand{\push}{{\sc{push}}}
\newcommand{\pull}{{\sc{pull}}}

\newcommand{\ALG}{{\sc raes}}

\newcommand{\bigO}{O}
\newcommand{\polylog}{\mathrm{polylog}}

\newcommand{\OLD}{\mathrm{OLD}} 
\newcommand{\GOOD}{\mathrm{GOOD}} 
\newcommand{\BAD}{\mathrm{BAD}} 

\newcommand{\BSDG}{\mathcal{TSG}}

\newcommand{\grad}{\mathrm{deg}}

\newcommand{\vol}{\text{vol}}
\newcommand{\age}{\mathrm{age}}
\newcommand{\bord}{\Gamma}

\title{Threshold-Driven Streaming Graph: \\ Expansion and Rumor Spreading}

\author{Flora Angileri\thanks{University of Rome ``Tor Vergata'', Rome, Italy. E-mail: \texttt{flora.angileri@students.uniroma2.eu}} \and Andrea Clementi\thanks{University of Rome ``Tor Vergata'', Rome, Italy. E-mail: \texttt{clementi@mat.uniroma2.it}}
 \and Emanuele Natale\thanks{CNRS, Université Côte d'Azur, I3S, INRIA, Sophia Antipolis, France. E-Mail: \texttt{emanuele.natale@univ-cotedazur.fr}}
 \and Michele Salvi\thanks{University of Rome ``Tor Vergata'', Rome, Italy. E-Mail: \texttt{salvi@mat.uniroma2.it}}
\and Isabella Ziccardi\thanks{CNRS, Université Paris Cité, IRIF, Paris, France. E-Mail: \texttt{isabella.ziccardi@irif.fr}}}

\date{}

\begin{document}
\maketitle








\begin{abstract}
A randomized distributed algorithm called \ALG{} was introduced in \cite{becchetti2020finding} to extract a bounded-degree expander from a dense $n$-vertex expander graph $G = (V, E)$. The algorithm relies on a simple threshold-based procedure. A key assumption in \cite{becchetti2020finding} is that the input graph $G$ is static~-- i.e., both its vertex set $V$ and edge set $E$ remain unchanged throughout the process~-- while the analysis of \ALG{} in dynamic models is left as a major open question.

In this work, we investigate the behavior of \ALG{} under a dynamic graph model induced by a \emph{streaming node-churn process} (also known as the \emph{sliding window model}), where, at each discrete round, a new node joins the graph and the oldest node departs. This process yields a bounded-degree dynamic graph $\mathcal{G} =\{ G_t = (V_t, E_t) : t \in \mathbb{N}\}$ that captures essential characteristics of peer-to-peer networks -- specifically, node churn and threshold on the number of connections each node can manage.
We prove that every snapshot $G_t$ in the dynamic graph sequence has good expansion properties with high probability.
Furthermore, we leverage this property to establish a logarithmic upper bound on the completion time of the well-known \push  \ and \pull\ rumor spreading protocols over the dynamic graph $\mathcal{G}$.
\end{abstract}

\section{Introduction}
\label{sec:intro}

In \cite{becchetti2020finding}, the authors proposed a  simple, lightweight   distributed algorithm, working on any  synchronous communication model,   that  extracts an $n$-vertex  sparse expander subgraph  from any $n$-vertex dense expander graph $G$. This task,  in different versions,  has been the subject of  a  strong      research activity \cite{allen2016expanders,ClementiNZ21,bansal2019new,giakkoupis2022,LS03,guptapand2025}.
The algorithm,  called \ALG{},\footnote{Standing for ``Request a link, then Accept if Enough Space''.}
  is governed by two parameters $c,d \in \mathbb{N} $ that essentially determine a constant threshold on  the maximum vertex degree, and it can be informally described as follows. Initially, each vertex has no incident \textit{links}.  In each round, every vertex  $v$ performs  two consecutive actions. In a first request phase, $v$ samples a set of random neighbors from the underlying graph $G$, selecting enough candidates to potentially establish $d$ \textit{outgoing} links. It then sends a link request to each of these sampled neighbors. In a second acceptance phase, each vertex, upon receiving requests,  accepts or rejects them based on a threshold rule. Specifically, it accepts all incoming requests from the current round unless doing so would result in more than $cd$ total \textit{incoming} links. If that limit is exceeded, it rejects all requests received in that round.
  The process repeats until every vertex has exactly $d$ established outgoing links, at which point the algorithm terminates and no further requests are made. 
Informally, in  \cite{becchetti2020finding} it is shown that, if the underlying graph $G$ from which each vertex selects its random neighbors is sufficiently dense\footnote{In particular, if the edge set has size $\Omega(n^2)$.} and has good expansion properties, 
then \ALG{} has $O(\log n)$ completion time and the subgraph determined by all the accepted links\footnote{In the final random subgraph produced by \ALG{}, both outgoing links and incoming ones  are  considered undirected.} is a good sparse expander, with high probability.\footnote{An event $E$ holds \emph{with high probability} (for short, w.h.p.) if $\Prob{E} \geq 1-n^{-\gamma}$ for some constant $\gamma>0$, with respect to some input parameter $n$.}

The setting considered in  \cite{becchetti2020finding} is  static: both the set of vertices and the underlying dense graph remain unchanged throughout the process. The work \cite{becchetti2020finding}  in fact leaves the analysis of \ALG{} in dynamic models as a major open question. This is motivated by the fact that modern network scenarios, such as peer-to-peer networks \cite{APRU12,nakamoto2008bitcoin,stutzbach2006understanding} and opportunistic networks \cite{boukerche2014opportunistic}, are inherently dynamic, with nodes and links changing over time, sometimes  at a relatively-high rate.

In more recent studies \cite{CPTZ21,becchetti2023expansion},  a different version of \ALG{} is presented and analyzed   over a dynamic setting where vertices may enter and leave the system according to the  \textit{streaming node-churn process},\footnote{They also considered other node-churn processes: we will discuss them in \Cref{sec:related}.} considered also in \cite{CooperKR08}. 
Despite its simplicity, this streaming model has been shown to be predictive for other, more realistic dynamic-graph models (see \cite{CPTZ21})
and, moreover, its rigorous analysis requires to cope with  challenging technical issues, as shown in  \cite{CPTZ21,CooperKR08} and for other  graph-connectivity problems  in \cite{crouch2013dynamic}. 
In this streaming model, starting from an empty vertex set $V_0$, at each round, a new vertex $v$ joins   the network and selects $d$ random neighbors. Then after $n$ rounds,  $v$ leaves the network and all its incident edges are removed. Notice that this   process implies that every vertex stays in the system for exactly $n$ rounds and, after an initial  time window of   $n$ rounds,  the number of alive vertices $|V_t|$ at every round $t$  is always $n$.  During its life, a vertex $v$ can thus see one of its incident link disappear because one of its neighbors is the oldest one and leaves the network: in that case, $v$ immediately replaces it with a new random link.

We remark that  the   dynamic version considered in \cite{becchetti2023expansion}  does not implement the second action of the original algorithm  \ALG{}: every link request is accepted by every destination vertex at any round of the process.
The absence of this second action clearly implies that the maximum vertex degree of the resulting dynamic graph is not bounded. Indeed, a standard balls-into-bins argument shows  that  the maximum degree  is  $\Theta(\log n/\log\log n)$, w.h.p.~(see for instance \cite{mitzenmacher2017probability}).
In the most   relevant network scenarios that inspired our algorithmic study, namely   peer-to-peer networks such as the the  bitcoin network \cite{baumann2014exploring,nakamoto2008bitcoin}, the  presence of an unbounded number of links managed by 
a single vertex may lead to serious efficiency and security problems \cite{albrecht2024larger,cruciani2023dynamic}. Indeed, the standard protocol of the bitcoin network \cite{cruciani2023dynamic,nakamoto2008bitcoin} imposes  a threshold on the number of active links each vertex can manage, thus    
an action similar to the second one of  the original \ALG{} algorithm proposed in \cite{becchetti2020finding}. For more discussion on this issue and other related works see \Cref{sec:related}.

A further motivation for maintaining dynamic bounded-degree expanders  lies in the opportunity to adopt   broadcast  protocols, such as \push\ and \pull\ ones, to  get fast and communication-efficient \textit{rumor  spreading} \cite{chierich_rumspread,demers1987epidemic,clementi2016rumor}.


As we discuss in the next subsection, our goal is to study the dynamic graph generated by the  original version of \ALG{} combined with the  streaming node-churn model. 


\subsection{Our contribution} \label{ssec:our}

\paragraph{Setting the dynamic-graph process.} We aim to analyze a dynamic graph model that simultaneously captures two key features of modern peer-to-peer networks: a local threshold mechanism that bounds the degree of each vertex, and a node-churn process that regulates how vertices join and leave the network in each round.
We kept all other modeling choices as simple and natural as possible, using the fewest parameters necessary. While this setting does not capture all aspects of real dynamic networks (such as the Bitcoin one), we believe that this approach can still recover qualitative properties and phenomena yielded by the simultaneous presence of the two features above, and that it can be robust to variations  or extensions of the model's complexity.

We introduce the \textit{Threshold-driven Streaming Graph} model, abbreviated as $\BSDG(n,d,c)$, which is obtained by combining the two processes described above: (i) the streaming node-churn model    \cite{becchetti2023expansion,cooper2007sampling}, and (ii)  the original \ALG{} protocol in \cite{becchetti2020finding}     (see \Cref{def:streaming-node-churn} and \Cref{def:edge-process} for its formal   definition). We first notice   that, in every round $t\geq 0$, the degree of each vertex  $v$     is always bounded by the  threshold $(c+1)d$: in particular, at most $d$ edges are generated by the  requests sent by $v$ and  at most $cd$ edges are due to the online requests received by $v$.

Consistently with 
other models of dynamic graphs with node churn \cite{augustine2015enabling,augustine2016distributed,guptapand2025,LS03},  we assume  the presence of a \textit{link manager} to apply the  \ALG's connection-request strategy: any  vertex that makes a link request    can access this entity and get a random destination vertex. Importantly enough, the role  of the link manager we assume here  is   minimal: vertices cannot get any   further information  from it.\footnote{For instance, one vertex  might ask the current degree of the selected destination or, even more, information about the current topology: this is not allowed.}
As we will elaborate later in this section,  the total number of calls each vertex performs to the link manager is a key performance measure of the system and  the \ALG's strategy  optimizes it.

As we will discuss later in Subsection \ref{ssec:our},      the $\BSDG(n,d,c)$  model yields a complex  stochastic process  of graph snapshots $\mathcal{G} = \{G_t = (V_t , E_t ): t \in \mathbb{N}\}$,
where     edges in $E_t$ are neither uniformly distributed nor    mutually independent.
Hence, the analysis of the key aspects, such as the expansion properties of the graph snapshots,      requires coping with  new technical issues that are likely to emerge in other, more realistic models as well.

\paragraph{Expansion properties.}
Even though  the node churn and the \ALG{} rules  are simple in themselves, their combination, yielding   the $\BSDG$ dynamic graph, turns out to be rather complex, essentially  because it generates both a non-uniform link  distribution and induces subtle correlations between the links of every snapshot of the dynamic graph.
Informally, on the one hand older vertices tend to have a higher degree than younger ones. On the other hand, the fact that connection requests might create conflicts with other requests and get rejected  several times along their life generates non trivial correlations among the links that are active in a given graph snapshot, even if they have been established in different previous rounds. 

Our analysis solves  the above technical challenges  and  essentially limits   the  maximum (i.e. worst-case) correlation lying among any subset of links of the  same snapshot (see \Cref{sec:overv} for an  overview of  this key technical part).   We then use such limited correlation among edges to prove that    the $\BSDG(n,d,c)$ model generates graph snapshots having the following good expansion properties.

\begin{theorem} [Expansion Properties] \label{thm:mainintro}
There exist constants $c$, $d$ and $\beta$ sufficiently large such that, for all $n$ large enough, and any round $ t \geq 2n$, the snapshot  $G_t$ generated   by    $\BSDG(n,d,c)$   has the following properties w.h.p.:
    \begin{enumerate}[(a)]
    \item There exists an induced  \textit{expander} subgraph in $G_t$ with  $n - O(\log n)$  nodes;
    \item Any    subset  of vertices of    size at least $\beta \log n$ has constant conductance.\footnote{For a definition of conductance see \eqref{eq:conductance}.} 
    \end{enumerate}
 
\end{theorem}

We observe that the above result is tight in the following sense. It is easy to see that a new incoming vertex may stay isolated for the first $o(\log n)$ rounds of its life with non negligible probability: then, it is clear that, at any round,  there may be some vertex subset of $o(\log n)$ size having bad expansion.

\paragraph{Communication costs.}
We  prove that, at every round $t \geq 0$, the overall number of calls to the link manager performed by the vertices in $V_t$ (i.e. the overall number of \textit{pending requests} at round $t$)  has constant expectation and is    $O(\log n)$, w.h.p. 
We  also show  that the overall number of calls each vertex  makes during all of its life has constant expectation and it is  $O(\log n)$, w.h.p, as well. These results are easy  consequences of    \Cref{lem:number-times-pending} and \Cref{lem:size-queue}.
\textit{Message-communication  overhead} is a crucial  performance parameter in communication networks  since it has a strong impact on   node traffic congestion and on the time delay of  fundamental tasks such as broadcast and consensus \cite{albrecht2024larger,augustine2016distributed,cruciani2023dynamic}. As for this aspect,  we observe that,   in  the $\BSDG(n,d,c)$ model,   the only messages exchanged by   vertices  are those determined  by   the pending link requests: our  bounds above therefore guarantee 
that the overall number of exchanged messages at  every round $t$ is optimal in expectation and $O(\log n)$, w.h.p. The same   bounds holds for the total number of messages (i.e.~the \textit{work}) every vertex exchanges during all of its life.

\paragraph{Rumor spreading.}
\textit{Rumor Spreading} is a class of simple  epidemic protocols that,   given  a source vertex $s$ holding a piece of information (i.e.~the \textit{rumor}), aim to  broadcast this information  to  all vertices of the graph.   
The basic, popular randomized  variants of rumor spreading are the  (synchronous) uniform \push\ protocol and the \pull\ protocol: in the former, at each round every informed node (i.e.,
every node that learned the rumor in a previous round)
chooses a neighbor uniformly at random and sends the
rumor to it. In  \pull ,   at each
round, every uninformed node chooses a random neighbor; if that neighbor is informed,  it sends the rumor to
the uniformed node. Finally, the \push -\pull\  protocol
combines  both  strategies above to inform new, uninformed nodes. 

\push\ and \pull\ protocols have been shown to be effective in many networks applications \cite{demers1987epidemic,harchol1999resource,van1998gossip}, and, very importantly for our setting, they have been proved to be fault-tolerant \cite{elsasser2009cover,feige1990randomized} and efficient even in some model of evolving graphs \cite{clementi2013rumor,clementi2016rumor,dutta2011information,giakkoupis2014randomized}.     A key   question      concerns the completion  time, i.e., how many rounds such protocols  take to broadcast  the source information to all nodes in the graph \cite{chierich_rumspread,karp2000randomized}. 

While   \textit{flooding} has been analyzed  even on dynamic graphs that include node-churn \cite{augustine2016distributed,becchetti2023expansion},  to the best of our knowledge, no analytical results are known for any rumor-spreading protocol.
As  a further contribution, we study the completion time of the uniform \push\ and \pull\ over  the  $\BSDG$ model and prove the following bound.

\begin{theorem} \label{thm:gossip-intro}
There exist constants $c$ and $d$  sufficiently large such that, for all $n$ large enough, the following holds. Let  $s$ be a \textit{source} node joining the  $\BSDG(n,d,c)$  dynamic graph at  some round $t_s \geq 2n$.   
Then, after $T= O(\log n)$ rounds, the \push~or the \pull~protocol inform at least $n-O(\log n)$ vertices in $G_{T+t_s}$, w.h.p.
\end{theorem}

Also the result of \cref{thm:gossip-intro} is tight for the same reasons of
\cref{thm:mainintro}: with non-negligible probability a new incoming vertex may stay isolated for the first $o(\log n)$ rounds and hence it cannot receive the source information. Then, it is clear that, at any round,  there may be some subset of size $o(\log n)$ with vertices that are not informed.

\subsection{Roadmap}
The rest of the paper is organized as follows.
In \Cref{sec:overv}, we   overview   the main technical challenges and the key ideas we introduce  to face   them.
In \Cref{sec:prely}, we provide all preliminaries required to formalize and study the $\BSDG$ model.   In \Cref{sec:keyfirst}, we give the first technical results on the link distribution generated by the $\BSDG$ model and, in particular, the key  \Cref{lem:multiple-requests-destination} that bounds the maximal correlation among multiple links of any graph snapshot. Then, in \Cref{sec:expans}, we describe how to use these  results   to prove the expansion properties stated in \Cref{thm:mainintro}.
 \cref{sec:pushpull} is devoted to the proof of the rumor spreading result, namely \cref{thm:gossip-intro}.
A further discussion on the motivations behind our research and a comparison with related works is provided in \Cref{sec:related}.
In \Cref{sec:concl}, we discuss some open questions.  Finally, some technical tools  are given in \Cref{app:prob}.

\section{Technical Analysis: An Overview} \label{sec:overv}

As we already remarked in \cref{sec:intro}, our analysis requires to cope with two main technical challenges, each   one already faced in     two previous works  \cite{becchetti2020finding} and \cite{becchetti2023expansion}  that analyze two different variants of \ALG{}.       
Unlike those prior works,  in which  only one of the two challenges is considered, our setting requires   to confront both simultaneously, significantly increasing the complexity of the analysis.

The first challenge, addressed in \cite{becchetti2020finding}, arises from the second action of \ALG{}, which involves the threshold-based conditional acceptance rule of link requests. This mechanism introduces correlations among the random   destinations of the accepted links: to see just one source of this correlation, consider   the fact  the acceptance of a link implies that the target node did not receive more than $cd$ requests in the current round.
In the static setting, \cite{becchetti2020finding} addresses this issue using a sophisticated compression argument to prove the expansion properties of the resulting graph. Essentially, while powerful, this technique lacks enough flexibility to include the presence of the second challenge: the node churn and the    dynamic link regeneration at every round. 

The second challenge thus arises from the presence of the  streaming node churn: this issue is  faced 
in \cite{becchetti2023expansion}, where   a simplified version of the \ALG{} algorithm is considered. In \cite{becchetti2023expansion},    vertices accept all incoming requests unconditionally, eliminating the threshold mechanism. This simplification avoids the correlation issues seen in the static case, allowing the authors to sidestep the compression argument. Their  proof relies on a key lemma establishing that the random destinations of the link requests follows an almost-uniform distribution; this property  is then exploited to get  good  expansion properties of the resulting graph snapshots.
 A major issue in their dynamic model is handling correlations due to node churn, especially proving that nodes with  similar ages do not generate dense clusters. On the other hand,  their key lemma may focus on the distribution of the destination of a \textit{single} link request: this  is enough since, in the absence of the threshold mechanism, link  destinations always remain mutually independent and their joint distribution is just a product. 
 In contrast, in our model, the threshold-based acceptance rule introduces dependencies among edge destinations: we thus have to cope with    both potential node clustering    \textit{and}  the mutual correlation among link destinations.

We address these issues 
by extending the approach of the key lemma from \cite{becchetti2023expansion}. Specifically, our \cref{lem:multiple-requests-destination} shows that,  not only the destination of a single link destination is   almost uniform (similarly to \cite{becchetti2023expansion}), but also demonstrates that the \textit{joint distribution} of the destinations of any subset of  links  can be effectively expressed as a product distribution, up to a constant factor. 

The main idea behind the proof of \cref{lem:multiple-requests-destination} can be informally summarized as follows.
Consider  the snapshot $G_t =(V_t,E_t)$   at round  $t \geq 2n$ 
and a set of link requests $R$. We want to control the probability that all requests in $R$ established a link to some set $P$ of vertices  at round $t$. As a first step, we order the requests according to the last time they were accepted by some node of $P$. This way, we can telescopically condition the probability that a single request $r\in R$ establishes a link with $P$ at some time $s$ on an event involving only connections happened in the past, see eq.~\eqref{eq:term-product-chain-requests}. For $r$ to connect to $P$ at time $s$ two events must happen: $r$ has to be pending at time $s$ and the link manager has to point to some node of $P$ at time $s$. Since we are conditioning only on the past, the probability of the second event is uniform over all nodes present at time $s$. 
As a byproduct, 
we are left to show that the (conditional) probability that $r$ is pending at time $s$ is small enough. The conditioning forces us to go through a (painful) worst-case scenario analysis, cf.~\eqref{eq:def-a-requests}. The key idea is the following: during its life each request goes through cycles (called $W_0,W_1,\dots$ in the proof) composed of two phases: a first phase where the request stays linked to a single vertex (until that vertex dies) and a second phase where the request is pending because it gets rejected before forming a new link. We show that, regardless of what happened in the past, the length of the first phase can be stochastically dominated from \textit{below} by a suitable uniform random variable, while the length of the second phase can be stochastically dominated from \textit{above} by a geometric random variable. The decomposition in cycles and the stochastic domination of the phases allow us to sandwich the event that $r$ is pending during its $f$-th cycle between two events (called $S_1(f)$ and $S_2(f)$ in the proof), see \eqref{eq:inequality-successful}. As $f$ varies, $S_1(f)$ and $S_2(f)$ form a partition of the space of events, allowing us to conclude. We believe that this technique can be also adapted to more complicated versions of node churn, as the Poisson node churn considered in different papers \cite{pandurangan2003building,becchetti2023expansion}.


Another technical challenge that lies behind all our proofs is the control of the number of pending requests at every round. This boils down to a queuing theory problem: thanks to the method of bounded differences, we can show that the process $(Q_t)_{t\in\mathbb N}$ of the number of pending requests can be stochastically dominated by a Markov process that has a strong negative bias for high values of $Q_t$ (\cref{lem:decreasing-queue}). This ensures that the queue of pending requests is $O(\log n)$ with high probability (\cref{lem:size-queue}). We also show in \cref{lem:number-times-pending} that the probability that a request is pending for more than $j$ rounds during its life decays exponentially, guaranteeing a minimal number of requests  to the  link manager and, thus, a minimal workload per node.

Given our key \cref{lem:multiple-requests-destination} and the control of the pending requests queue, the proof of the good expansion properties of the dynamic graph become more standard, albeit suitable adaptations of the techniques of \cite{becchetti2023expansion,becchetti2020finding} are needed in our framework.

Finally, the expansion properties of the dynamic graph and the fact that our model allows by its nature only vertices of bounded degree would make the results of \Cref{thm:gossip-intro} a simple consequence of the classic  analysis of rumor spreading in \cite{chierich_rumspread}. The  only novelty here is the analysis of the  initial \textit{bootstrap} process.  The bootstrap  of the information-spreading process is essentially the initial, random  time phase    the protocol requires to reach a logarithmic number of informed nodes: we need this further analysis since  \Cref{thm:mainintro}  does not guarantee worst-case good expansion   for subsets of informed vertices of size $o(\log n)$. Informally, for this phase, we   use Claim (b) of our \Cref{thm:mainintro} to prove that, when joining the graph,  the source has high  probability  to fall into a connected component  of size $\Omega(\log n)$ and, moreover, this  component will be stable for at least $\Theta(\log^2 n)$ rounds. This is enough to get $\Theta(\log n)$ number of informed nodes after a logarithmic number of rounds after the source  joined the graph. As remarked above, once the set of informed nodes achieves a logarithmic size, we can combine Claim (a) of \cref{thm:main-expansion} with the previous classic analysis of rumor spreading in \cite{chierich_rumspread} to get \Cref{thm:gossip-intro}.

\section{Preliminaries} \label{sec:prely}


A \emph{dynamic graph} $\mathcal{G}$ is an infinite sequence of graphs $\mathcal{G} = \{G_t = (V_t , E_t ): t \in \mathbb{N}\}$. If $\{V_t\}_t$ or $\{E_t\}_t$ are sequences of random sets, we call the corresponding random process a \emph{dynamic random graph}, and   $G_t$ denotes the \emph{snapshot} of the dynamic graph at \textit{round}  $t$. As usual, the size of any subset $A$ is denoted as $|A|$.
The \textit{outer boundary} of a  set of vertices $S$ is defined as 
\[ \bord_t(S) = \{v \in V_{t} \setminus S \ | \ \exists u \in S \ \mbox{s.t.} \ \{u,v\} \in E_t   \}\,. \]

Our analysis of dynamic graphs considers the fundamental notions of \textit{conductance} of a graph \cite{hoory2006expander}.
For any two set of vertices $S,T \subseteq V_t$,   $E_t(S,T)$ denotes the set of edges crossing $(S,T)$ at round $t$, that is
$E_t(S,T) = \{\{u,v\} \in E_t : u \in S, v \in T\}$,
while     $\partial_t S = E_t(S,V_t \setminus S)$ 
denotes  the set of edges crossing $(S,V_t \setminus S)$. The \emph{volume} of the set $S$ is defined as $\vol_t(S) = |E_t(S,V_t)|$.
Then,   the \emph{conductance} $\phi_t(S)$  of the set $S$ at round  $t$ is defined as  
\begin{equation}
  \phi_t(S) = \frac{|\partial_t S|}{\min\{\vol_t(S), \vol_t(V_t \setminus S)\}}. 
  \label{eq:def-conductance}
\end{equation}

The conductance of the graph $G_t$ is the minimum of $\phi_t(S)$ over all possible sets $S \subseteq V_t$ with volume smaller than the total number of edges:
\begin{equation}
  \phi_t(G_t) = \min_{\substack{S\subseteq V_t}}\phi_t(S). 
  \label{eq:conductance}
\end{equation}
Given any vertex  subset $S$,   $G_t[S]$ denotes   the subgraph of $G_t$ induced by  $S$.
We will omit the subscript $t$ in  all   notations above when it is clear from the context.  


\begin{definition}[Graph Expansion]
An infinite family of graphs $\{G^{(n)}(V,E), \mbox{ with } |V| =n \}_{n \in \mathbb{N}}$   is
an  $\alpha$-\emph{expander} if there exist  constants  $\alpha\in(0,1)$ and $n_0 \in \mathbb{N}$  such that   $\phi(G^{(n)}) \geq \alpha$ for all  $n \geq n_0$. 
\label{def:graph-expansion}
\end{definition}

\subsection{The dynamic graph model}
Our goal is to study the   dynamic graph model  determined by combining the streaming node-churn process \cite{becchetti2023expansion} with the edge generation process defined by the distributed algorithm \ALG 
~(in \cite{becchetti2020finding}), based on a simple threshold rule. In what follows, we formalize this combined model and  state some of its  preliminary properties.

The vertex-set process   $\{V_t\}_t$ of a dynamic graph  $\mathcal{G}$ is typically called \emph{node churn} \cite{augustine2016distributed,becchetti2023expansion}. In this paper, we consider  the  deterministic \textit{streaming}   node churn  of  parameter  $n$ defined as follows.

\begin{definition}[Streaming node churn] Let $n \in \mathbb{N}$. A \emph{streaming node churn} with $n$ vertices   is a deterministic process $\{V_t: t \in \mathbb{N}\}$ such that $V_0 = \emptyset$, and, for any $t \geq 1$, the set $V_{t}$ is defined iteratively by the following simple rules: 
\begin{enumerate}[(a)]
    \item A new vertex $v$ joins the vertices   set; 
    \item At round $t \geq n+1$, the vertex $u$   that joined the set of vertices   at time $t-n$,  leaves the graph.
\end{enumerate} 
Then, $V_{t}$ is defined to be  $V_{t} = V_{t-1} \cup \{v\} \setminus \{u\}$ when $t \geq n+1$ and $V_t = V_{t-1} \cup \{v\}$ for $t\leq n$.   For a vertex  $v \in V_t$,    the  \textit{age} of     $v$ at time $t$  is the function $\age_t(v) = t-t_v$, where $t_v \leq t$ is the round vertex $v$ joined the vertex  set.  
\label{def:streaming-node-churn}
\end{definition}

Some easy but important remarks follow.
The vertex  $v$ joining the graph at time $t_v$ leaves the graph at round  $t_v+n$, i.e. $v \in \cap_{s=t_v}^{t_v+n-1}V_s$ and $v \not \in V_{t+n}$. 
We say that the streaming node churn  $\{V_t: t \in \mathbb{N}\}$ with parameter $n$   gets \emph{stable} after round $t \geq 2n$: in particular, after that round,   two properties hold that we will often (implicitly)  use  in the analysis of the process:
\begin{enumerate}[(i)]
    \item The set $V_t$ has size $n$;
    \item The set $V_{t-n}$ has size $n$: this implies that, at the round each vertex in $V_t$ joined the graph, there were already $n$ vertices present in the graph.
\end{enumerate}

In order to define our  dynamic graph model $\mathcal{G}$, we need also to specify the evolution of the edge set $\{E_t\}_t$. We consider a random process $\{E_t\}_t$ determined by the  simple rules of the \ALG\ algorithm we described in \Cref{sec:intro}. 
According to  peer-to-peer   models (where vertices make \textit{connection requests} to other nodes), we distinguish between  \textit{outgoing} edges from a vertex $v$, originating from a connection request made by $v$, and  \textit{incoming} edges to $v$, resulting from a connection request made by another vertex to $v$.
However,  we remark that the resulting graph snapshots $G_t=(V_t,E_t)$ are \textit{undirected}: once established, every edge  in $E_t$  allows message communication  in both directions.

\begin{definition}[Edge process]
\label{def:edge-process}
Let  $c,d \in \mathbb{N}$ be two parameters, and let $\{V_t:t \in \mathbb{N}\}$ be the streaming node churn   with $n \geq 2$ vertices introduced  in \Cref{def:streaming-node-churn}.
The random subset sequence $\{E_t\}_t$ is defined inductively as follows. We set  $E_0 = \emptyset$ and, for any $t \geq 1$, the subset $E_{t}$ is generated
according to the following  rules:\footnote{Essentially, each round $t$ is organized in two consecutive phases: in the first one, the node churn  action is applied to $V_{t-1}$ thus getting $V_t$, while, in the second phase,  the edge process works on the   new vertex subset $V_t$.}
\begin{enumerate}[(a)]
    \item $E_{t}$ contains all the edges in $E_{t-1}(V_{t},V_{t})$, while   all     edges incident to the leaving vertex of age $n$  are deleted;
    \item Each vertex $v \in V_{t}$ with less than $d$ \emph{outgoing} edges makes a new \textit{connection request} for each one of its  missing outgoing edges. Each request is sent to a destination vertex chosen independently and  uniformly at random in  $V_{t}\setminus \{v\}$.\footnote{Notice that this rule implies that the new node, when it joins the graph,  will make exactly $d$ connection requests.}
    \item Assume a  vertex $u \in V_{t}$ receives $\ell \geq 1$ connection requests from other nodes. Then, it accepts all the requests and activates the corresponding edges \textit{if and only if} it has \emph{in-degree} $ \leq  c \cdot d -\ell$; otherwise, it rejects \textit{all} the   requests it received at round $t$.
\end{enumerate}
\end{definition}
 
Informally, each vertex  $v \in V_t$ of the dynamic graph tries to maintain its  out-degree equal to $d$: we can think that $v$ is equipped with $d$ connection requests that it tries to keep connected to active vertices. 
However, if a request of $v$ at time $t$   lands to a vertex $u$ which has a number of incoming edges and new connection requests  larger than  $c d$,   the request of $v$ is rejected and will not create an edge at round $t$ (but it will try to connect again at the next round).

The dynamic graph $\mathcal{G}$ determined by the streaming node churn  in \cref{def:streaming-node-churn} and the edge process in \cref{def:edge-process} will be called \textit{Threshold-driven Streaming Graph} with parameters $n$, $d$, and $c$ (for short $\BSDG(n,d,c)$).

\paragraph{Full nodes, pending requests, and other key
 random variables.} 
We now introduce  the key notions  and quantities we will consider in the probabilistic analysis of  the $\BSDG(n,d,c)$ model.

A vertex with $cd$ incoming edges is called  \emph{full} and  
the set of full vertices at round $t$ is denoted as $B_t$.

Each request at round  $t$ is a pair $r= (v,i)$, where $v \in V_t$ is the vertex making the request and $i \in [d]$ is its index.
For any vertex $v \in V_t$ (or any request $r \in V_t \times [d]$), we will denote with $t_v$ (resp.~$t_r$) the first round in which   $v$ (resp.~$r$) appears in the dynamic graph.

If a request $r$ is trying to connect to some vertex $u$, we say that $r$ \emph{targets} the vertex $u$.


We observe that, at any round, there are \textit{pending} requests. A connection request $r$ from a vertex $v$ is called pending at round $t$ if either $v$ has just joined the set of vertices $V_t$, or if $r$ has been rejected in round $t-1$, or if $r$ was connected at round $t-1$ to the node $u$ that leaves the network at round $t$. Such a request     generates an edge in  $E_t$ if and only if it is  accepted by its target vertex at round $t$.
Notice that, when a vertex joins the graph at time $t$, all its $d$ requests are pending at time $t$.  

The  \emph{queue at round  $t$}  is the random set $Q_t$ of all pending requests at round $t$.
As we will see in the next sections, the queue plays a key role in our analysis. Moreover, by the definition of the   $\BSDG(n,d,c)$ model, the size $|Q_t|$ of the queue bounds the overall number of messages exchanged by the vertices at round $t$.

\begin{fact} \label{fact:msgcompl}
For any $t\geq 1$, the \textit{overall number of messages} performed by the dynamic graph $\BSDG$ at round $t$ is $O(|Q_t|)$.
\end{fact}




For any   $t \geq 1$ and any request $r \in V_t \times [d]$,  the random variable $X_t(r)$ is defined as the destination of the request $r$ if  $r$ is accepted (and thus generates an edge in $E_t$), while we set $X_t(r) = \emptyset$ if the request $r$ was rejected at round $t$. 

For any set $S$, denote with  $r \xrightarrow{t} S$ the event indicating that the request $r$ established a connection with a vertex in the set $S$ at  round $t$: in other words, that the request $r$ is pending at round $t$, targets a vertex in the set $S$ and it is accepted.

\paragraph{On the number of full vertices.}
 Using a simple combinatorial argument, we next prove that the  size   of the  set $B_t$ 
 of full vertices  (i.e. vertices with in-degree equal to $cd$) at round $t$  can never exceed a suitable threshold.  

 \begin{claim}
 \label{claim:bound-overloaded-nodes}
     For any $t \geq 1$, $|B_t| \leq \tfrac{n}{c}$.
 \end{claim}

\begin{proof}
For each $t \geq 1$, it holds $|E_t| \leq nd$, since each vertex has at most $d$ outgoing edges.
Assume, by contradiction, that $|B_t| > \frac{n}{c}$. Then, since each vertex in $B_t$ has in-degree $cd$, this implies that
$|E_t| \geq c d |B_t| > cd \cdot \tfrac{n}{c} = nd$,
contradicting the fact that $|E_t| \leq nd$.
\end{proof}

\section{Key Lemmas} \label{sec:keyfirst}
In this section we provide an analysis of    the 
stochastic process generated by the $\BSDG$ model. This analysis   allows us to establish  some key results that will be then used to derive the expansion properties claimed in \Cref{thm:mainintro} and the logarithmic bound on the completion time of the \push \ and \pull \ protocols in \Cref{thm:gossip-intro}.

\subsection{On the  edge probability distribution }

To analyze the expansion properties of the $\BSDG$ snapshots, we show that the link requests from any subset of nodes are both nearly uniformly distributed across the entire node set and nearly mutually independent. 
This result is the main technical contribution of the paper and is formalized in the following.
 
\begin{lemma}
There exist constants $c$ and $d$ sufficiently large such that, for all $n$ large enough, the following holds. 
    For every $t \geq 2n$, and for every $S \subseteq V_t$, $R \subseteq S \times [d]$ and $P \subseteq V_t$, we have
    \[\Prob{\cap_{r \in R}\{X_t(r) \in P\}} \leq \left(\frac{220|P|}{n-1}\right)^{|R|}.\]
    \label{lem:multiple-requests-destination}
\end{lemma}

\begin{proof}
Let $k = |R|$ and denote the requests in $R$ with $r_1,\dots, r_k$. Recall that $\{r \xrightarrow{t} P\}$ is the event indicating that the request $r$ is pending at round $t$, targets a node in $v \in P$, and it is accepted by $v$.
Recall that, for each request $r$, $t_r$ indicates the round the request joins the graph. 

\medskip

\noindent {\textbf {STEP 1: Conditioning on past events.}}

\smallskip

We can decompose
\begin{align*}
    \Prob{\cap_{j=1}^{k}\{X_t(r_j) \in P\}} = \sum_{s_1,\dots, s_k} \Prob{\cap_{j=1}^{k} \{r_j \xrightarrow{s_j} P\}},
\end{align*}
where the sum is taken over $s_j \in \{t_j,\dots, t\}$.
For any sequence of rounds $\mathbf{s}=(s_1,\dots, s_k)$, denote with $i_1(\mathbf{s}),\dots, i_k(\mathbf{s})$ the requests such that the times $(s_{i_j(\mathbf{s})})_j$ are in increasing order, i.e.
\[s_{i_1(\mathbf{s})} \leq s_{i_2(\mathbf{s})}\leq \cdots \leq s_{i_k(\mathbf{s})}.\]
For simplicity, we will denote $i_j(\mathbf{s})$ with $i_j$, but notice that the values of $i_1,\dots, i_k$ depends on the sequence $\mathbf{s}$. Then, we have that
\begin{align}
 \notag   \Prob{\cap_{j=1}^k \{X_t(r_j) \in P\}} &= \sum_{s_1,\dots, s_k}
 \Prob{\cap_{j=1}^k \{r_{i_j} \xrightarrow{s_{i_j}} P\}} 
    \\ & = \sum_{s_1,\dots, s_k} \prod_{j=1}^k \Prob{r_{i_j} \xrightarrow{s_{i_j}} P \,\big| \cap_{h=1}^{j-1}\{r_{i_h} \xrightarrow{s_{i_h}} P\}}.
    \label{eq:term-product-chain-requests}
\end{align}
We need to examine each term in \eqref{eq:term-product-chain-requests}. We have that the event $\{r_{i_j} \xrightarrow{s_{i_j}} P\}$ is the intersection of the events $\{\text{$r_{i_j}$ targets $P$ at time $s_{i_j}$}\}$, $\{\text{$r_{i_j}$ is pending at time $s_{i_j}$}\}$ (which can also be written as $\{r_{i_j} \in Q_{s_{i_j}}\}$) and $\{\text{the request $r_{i_j}$ is accepted at time $s_{i_j}$}\}$. The target of $r_{i_j}$ at time $s_{i_j}$ is chosen uniformly at random in $V_{s_{i_j}}$: therefore, this is independent from the past and from the fact that $r_{i_j}$ is pending at time $s_{i_j}$.
Moreover, the request $r_{i_j}$ targets some node in $P$ at time $s_{i_j}$ with probability at most $\frac{|P|}{n-1}$ (since some nodes in $P$ may not be in $V_{s_j}$).
Hence,  for each $j \in [k]$,
\begin{align}
\notag
&\Prob{r_{i_j} \xrightarrow{s_{i_j}} P \mid \cap_{h=1}^{j-1}\{r_{i_j} \xrightarrow{s_{i_j}} P\}} 
\\ \notag &\leq \Prob{ r_{i_j} \in Q_{s_{i_j}} \text{ and $r_{i_j}$ targets $P$ at time $s_{i_j}$}
     \mid \cap_{h=1}^{j-1}\{r_{i_h} \xrightarrow{s_{i_h}} P\}} 
    \\ & \leq \frac{|P|}{n-1} \cdot \Prob{r_{i_j} \in Q_{s_{i_j}} \mid \cap_{h=1}^{j-1}\{r_{i_h} \xrightarrow{s_{i_h}} P\}}. \label{eq:requests-pending}
    \end{align}
For every fixed sequence of rounds $s_1,\dots, s_k$ with $s_\ell = s_{i_j}$, notice that:
\begin{align*}
    &\Prob{r_{i_j} \in Q_{s_{i_j}} \mid \cap_{h=1}^{j-1}\{r_{i_h}\xrightarrow{s_{i_h}} P\}} =\Prob{r_\ell \in Q_{s_{\ell}} \mid \cap_{h:s_h \leq s_\ell}\{r_{h}\xrightarrow{s_{h}} P\}}.
\end{align*}

\medskip

\noindent {\textbf {STEP 2: Conclusion of the proof assuming worst-case bound.}}

\smallskip

Consider now the following quantity
\begin{align}
    &A(r_\ell,s_\ell) = \max_{s_1,\dots, s_{\ell-1},s_{\ell+1},\dots s_k} \Prob{r_\ell \in Q_{s_{\ell}}\mid \cap_{h:s_h \leq s_\ell}\{r_{h}\xrightarrow{s_{h}} P\}}
    \label{eq:def-a-requests}
\end{align}
that is, $A(r_\ell,s_\ell)$ is the probability that $r_\ell$ is in the queue of rejected requests at time $s_\ell$ under the worst possible conditioning involving events of the type $\{r_h\xrightarrow{s_{h}} P\}$ with $s_h\leq s_\ell$.
Notice that the quantity in \cref{eq:term-product-chain-requests} can be bounded in terms of \eqref{eq:def-a-requests}: indeed, considering the previous bound and \eqref{eq:requests-pending}, it holds
\begin{align}
\label{eq:bound-aell}
    \prod_{j=1}^k \Prob{r_{i_j} \xrightarrow{s_{i_j}} P \mid \cap_{h=1}^{j-1}\{r_{i_h} \xrightarrow{s_{i_h}} P\}} \leq \left(\frac{|P|}{n-1}\right)^k \cdot \prod_{\ell = 1}^k A(r_\ell,s_\ell).
\end{align}
Notice that, if $A(r_\ell,s_\ell) \leq \tfrac{215}{n}+ {\rm e}^{-(s_\ell-t_\ell)}$, we proved the lemma. Indeed, from \eqref{eq:term-product-chain-requests} and \cref{eq:bound-aell} it holds that
\begin{align*}
    \Prob{\cap_{j=1}^k \{X_t(r_j) \in P\}} &\leq \left(\frac{|P|}{n-1}\right)^k \sum_{s_1,\dots,s_k} \prod_{\ell=1}^k\left(\frac{215}{n} + {\rm e}^{-(s_\ell-t_\ell)}\right)
    \\ & \leq \left(\frac{|P|}{n-1}\right)^k \left(\sum_{h = 0}^{n-1}\left(\frac{215}{n}+ {\rm e}^{-h}\right)\right)^k
    \\ & \leq \left(\frac{220|P|}{n-1}\right)^k.
\end{align*}
The rest of the proof is devoted to show that, for each $r_\ell$ and each $s_\ell$, it holds $A(r_\ell,s_\ell) \leq \tfrac{215}{n}+{\rm e}^{-(s_\ell-t_\ell)}$.

\medskip

\noindent {\textbf {STEP 3: Proof of the worst-case bound: setup.}}

\smallskip

Fix $\ell \in [k]$, and denote with $a_1,\dots,a_{\ell-1},a_{\ell + 1},\dots, a_k$ some rounds such that $a_i \in \{t_i,\dots, t\}$ and attaining the maximum in \eqref{eq:def-a-requests}, i.e.
\[A(r_\ell,s_\ell) = \Prob{r_\ell \in Q_{s_{\ell}} \mid \cap_{h:a_h \leq s_\ell}\{r_{h}\xrightarrow{a_{h}} P\}}.\]
Before proceeding with the proof, notice that $s_\ell \geq t_\ell$, and we can decompose the interval $\{t_\ell,\dots, s_\ell\}$ in sub-intervals $W_0,W_1,W_2,\dots$, that we define iteratively according to the behavior of the request $r_\ell$.
The initial point of $W_0$ is $w_0 = t_\ell$. 
The initial point $w_1$ of $W_1$ is such that
\[w_1 = \min\{s \geq t_\ell: X_s(r_\ell) \neq \emptyset \text{ or } s = \min\{t_\ell+n-1,t\}\},\]
and for each $i \geq 2$ the initial point $w_i$ of the interval $W_i$ is such that
\begin{equation}
  w_i = \min\{s \geq w_{i-1}: (X_s(r_\ell) \neq X_{w_{i-1}}(r_\ell) \text{ and $X_s(r_\ell) \neq \emptyset$}) \text{ or }s = \min\{t,t_\ell+n-1\}\}.  
  \label{eq:def-Wi}
\end{equation}
In other words, the interval $W_0$ contains the initial rounds when $r_\ell$ joins the graph but all its targets reject the requests ($W_0$ may also have length $0$). For $i \geq 1$, the interval $W_i$ contains the rounds in which the request $r_i$ is connected to the same vertex, which is in particular the $i$-th different destination of $r_i$ during its life, plus the rounds where $r_i$ is pending after its $i$-th destination leaves the network and before connecting to the $i+1$-th one. We define also $Z_0,Z_1,\dots$ as suitable sub-intervals of the intervals $W_0,W_1,\dots$  for which  the initial point $z_i$ of $Z_i$ is such that
\begin{equation}
  z_i = \min\{s \geq w_i: X_s(r_\ell) = \emptyset \text{ or } s=\min\{t,t_\ell+n-1\}\} \,   
  \label{eq:def-Zi}
\end{equation}
and the final point of $Z_i$ coincides with the final point of $W_i$. Notice that, from the definition of $W_0$, it holds $Z_0 = W_0$.
We also define $F$ as the last interval $W_F$ intersecting $\{t_\ell,\dots, s_\ell\}$. For a better understanding, see also \cref{fig:requests-r-ell}.

\begin{figure}
    \centering    \includegraphics[width=0.9\linewidth]{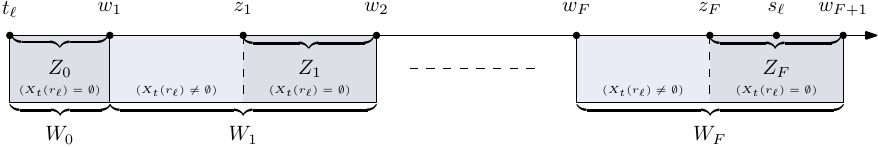}
    \caption{Variables $W_0,W_1,\dots, W_F$ and $Z_0,Z_1,\dots, Z_F$ in the time interval $\{t_\ell,\dots, s_\ell\}$ referring to the behavior of the request $r_\ell$.}
    \label{fig:requests-r-ell}
\end{figure}

Denote  for simplicity by $\mathcal C$ the event on which we are conditioning:
\[\mathcal{C} = \cap_{h: a_h \leq s_\ell}\{r_h \xrightarrow{a_h} P\}.\]
Notice that, for any $f \geq 1$, it holds $|W_0|+ |W_1| + \dots + |W_{f-1}| = w_f - t_\ell$, where $w_f$ is random. 

\medskip

\noindent {\textbf {STEP 4: Proof of the worst-case bound: comparison with $S_1$ and $S_2$.}}

\smallskip

Consider now the following event, for any fixed $f \geq 1$  and $\Tilde{w}_f \geq 0$
\begin{align*}
  \mathcal{E}(f,\Tilde{w}
  _f)  = \mathcal{C} \cap \{w_f = \Tilde{w}_f\}.
  \end{align*}
Notice that, for any $\Tilde{w}_f$, it holds
\begin{equation}
\Prob{s_\ell \in Z_f \mid \mathcal{E}(f,\Tilde{w}_f)} = \Prob{\Tilde{w}_f + |W_f \setminus Z_f| < s_\ell \leq \Tilde{w}_f + |W_f| \mid \mathcal{E}(f,\Tilde{w}_f)}\,.
    \label{eq:triplet-not-successful}
\end{equation}
On the other side, we consider two disjoint events $S_1(f)$ and $S_2(f)$ defined as follows
\begin{equation}
    S_1(f)= \{\text{$s_\ell \in W_f \setminus Z_f$, $|W_f\setminus Z_f|\geq \tfrac{7n}{8}$}\}
    \label{eq:triplet-successful-cond-1}
\end{equation}
and
\begin{equation}
S_2(f) = \{\text{$s_\ell \in W_{f+1}\setminus Z_{f+1}$, $|W_{f+1}\setminus Z_{f+1}| < \tfrac{7n}{8}$}\}.
\label{eq:triplet-successful-cond-2}
\end{equation}
We notice that, for any fixed $\Tilde{w}_f$, we have
\begin{equation}
  \Prob{S_1(f) \mid \mathcal{E}(f,\Tilde{w}_f)} = \Prob{|W_f\setminus Z_f| \geq \tfrac{7}{8}n\,,\; \Tilde{w}_f\leq s_\ell < \Tilde{w}_f + |W_f\setminus Z_f| \mid \mathcal{E}(f,\Tilde{w}_f)}  
  \label{eq:S1-given-E}
\end{equation}
and
\begin{align}
   &\notag \Prob{S_2(f) \mid \mathcal{E}(f,\Tilde{w}_f)} 
   \\ &= \Prob{|W_{f+1}\setminus Z_{f+1}|< \tfrac{7}{8}n\,,\; \Tilde{w}_f +  |W_f| \leq s_\ell < \Tilde{w}_f + 
 |W_f| + |W_{f+1}\setminus Z_{f+1}| \mid \mathcal{E}(f,\Tilde{w}_f)}.
 \label{eq:S2-given-E}
\end{align}

Conditioning on $\mathcal{E}(f,\Tilde{w}_f)$, the random variables $|W_f \setminus Z_f|$, $|W_{f+1}\setminus Z_{f+1}|$ and $|Z_f|$ have the following characteristics. The variables $|W_f \setminus Z_f|$ (resp. $|W_{f+1}\setminus Z_{f+1}|$) are determined as follows: in round $w_f$ (resp. $w_{f+1}$), the request $r_\ell$ is connecting to a not full vertex in the graph (indeed, from the definition of $W_f \setminus Z_f$ and $W_{f+1}\setminus Z_{f+1}$, it holds $X_{w_f}(r_\ell) \neq \emptyset$ and $X_{w_{f+1}}(r_\ell) \neq \emptyset$). 
Since $r_\ell$ is targeting a node uniformly at random in the graph, and since in round $w_f$ (resp. $w_{f+1}$) the request targets a \emph{not} full vertex, then the target of $r_\ell$ will be a uniform random not full vertex in the graph at round $w_f$ (resp. $w_{f+1}$) and the length of $|W_f \setminus Z_f|$ (resp. $|W_{f+1}\setminus Z_{f+1}|$) is $n$ minus the length of the life of the targeted node.  
Since from \cref{claim:bound-overloaded-nodes} the number of full vertices is at most $\tfrac{n}{c}$, the value of $|W_{f}\setminus Z_f|$ is a uniform random variable in a set $H \subseteq \{1,\dots, n\}$ with $|H| \geq n-\tfrac{n}{c}$.
Analogously,  the value of $|W_{f+1}\setminus Z_{f+1}|$ is a uniform random variable in a set $H' \subseteq \{1,\dots, n\}$ with $|H'| \geq n-\tfrac{n}{c}$.
The random variable $|Z_f|$, instead, from \cref{claim:bound-overloaded-nodes}, for every possible conditioning and value of $\mathcal{E}(f,\Tilde{w}_f)$ and $|W_f \setminus Z_f|$, is stochastically bounded by a geometric random variable with success parameter $1-\tfrac{1}{c}$.

Consider $\Tilde{w}_f$ such that $s_\ell-\Tilde{w}_f > \frac{101}{100}n$, then it holds that, since $|W_f \setminus Z_f| \leq n$ with probability $1$,
\begin{equation}
    \label{eq:bound-big-wf}
    \Prob{s_\ell \in Z_f \mid \mathcal{E}(f,\Tilde{w}_f)} \leq \Prob{|Z_f| \geq \tfrac{1}{100}n \mid \mathcal{E}(f,\Tilde{w}_f), |W_f\setminus Z_f|} \leq \left(\frac{1}{c}\right)^{\frac{1}{100}n} \leq {\rm e}^{-\frac{n}{100}}.
\end{equation}

We will show that, for any $\Tilde{w}_f$ such that $s_\ell-\Tilde{w}_f \leq \frac{101}{100}n$, it holds
\begin{align}
     &\Prob{s_\ell \in Z_f \mid \mathcal{E}(f,\Tilde{w}_{f})}
    \leq \frac{210}{n} \cdot \left(\Prob{S_1(f) \mid \mathcal{E}(f, \Tilde{w}_f)} + \Prob{S_2(f) \mid \mathcal{E}(f,\Tilde{w}_f)}  \right). 
    \label{eq:inequality-successful}
\end{align}

We notice that, when $s_\ell - \Tilde{w}_f < 0$, \eqref{eq:inequality-successful} follows trivially since $\{s_\ell \in Z_f\}$, given $\mathcal{E}(f,\Tilde{w}_f)$, has zero probability. We first notice that, from \eqref{eq:triplet-not-successful} that for each fixed $z = |Z_f|$ and $\mathcal{E}(f,\Tilde{w}_f)$ there are at most $z$ values of $w = |W_f \setminus Z_f|$ for which it holds $\{s_\ell \in Z_f\}$.
Hence, we have that


\begin{align}
\notag \Prob{s_\ell \in Z_f \mid \mathcal{E}(f,\Tilde{w}_f)}& = \sum_{z=0}^{+\infty}\Prob{z = |Z_f|\mid \mathcal{E}(f,\Tilde{w}_f)} \frac{z}{|H|} 
    \\ \notag &\leq \frac{1}{|H|}\Expc{|Z_f| \mid \mathcal{E}(f,\Tilde{w}_f)}
    \\ \notag &\leq \frac{1}{n-\tfrac{n}{c}}\Expc{\text{Geom}(1-\tfrac{1}{c})}
    \\   \label{eq:bound_NF}   &= \frac{1}{n(1-\tfrac{1}{c})^2} \leq \frac{2}{n} .
\end{align}

Now we consider two different cases, depending on the value of $\Tilde{w}_f$ in the range $0\leq s_\ell - \Tilde{w}_f  \leq \frac{101}{100}n$.

In the first case, we assume that $0\leq s_\ell - \Tilde{w}_f \leq \tfrac{n}{2}$. In such a case, from \eqref{eq:S1-given-E},  for each fixed $z = |Z_f|$, there are at at least
\[ n-\tfrac{n}{c} - \max \left\{\tfrac{7}{8}n, s_\ell-\Tilde{w}_f\right\} \geq  n \left(\tfrac{1}{8}-\tfrac{1}{c}\right)\]
values of $|W_f\setminus Z_f|$ in $H$ for which it holds $S_1(f)$. Notice that this holds independently from the value of $|Z_f|$. Hence,
\[\Prob{S_1(f) \mid \mathcal{E}(f,\Tilde{w}_f)}\geq n\left(\frac{1}{8}-\frac{1}{c}\right)\cdot \frac{1}{|H|} \geq \frac{1}{16}.\]
From the inequality above and from \eqref{eq:bound_NF}, we get  
\[\Prob{s_\ell \in Z_f \mid \mathcal{E}(f,\Tilde{w}_f)} \leq \frac{32}{n}\Prob{S_1(f) \mid \mathcal{E}(f,\Tilde{w}_f)},\]
which proves \eqref{eq:inequality-successful} in the first case.

We are left to analyze the second case, that is when $\frac{n}{2}<s_\ell - \Tilde{w}_f \leq \tfrac{101n}{100}$.  From \eqref{eq:S2-given-E}, for each fixed $z = |Z_f|$ ,   we have at least
\begin{equation}
  (s_\ell - \Tilde{w}_f-z-\tfrac{n}{4}-\tfrac{n}{c})\left(\tfrac{7}{8}n - (s_\ell-\Tilde{w}_f-z)+\tfrac{n}{4}-\tfrac{n}{c}\right)
  \label{eq:successful-case2}
\end{equation}
values for the pair $|W_f\setminus Z_f|$ in $H$ and $|W_{f+1}\setminus Z_{f+1}|$ in $H'$ such that $|W_{f+1}\setminus Z_{f+1}| < \tfrac{7}{8}n$ and $|W_{f}\setminus Z_f| \geq \tfrac{n}{4}$, and for which it holds $S_2(f)$. 
Since $\frac{n}{2}< s_\ell-\Tilde{w}_f \leq  \tfrac{101n}{100}$, 
such pairs are, for $c$ large enough, at least
\begin{align*}
     (s_\ell - \Tilde{w}_f-z-\tfrac{n}{4}-\tfrac{n}{c})\left(\tfrac{7}{8}n - (s_\ell-\Tilde{w}_f-z)+\tfrac{n}{4}-\tfrac{n}{c}\right) \geq \left( \tfrac{n}{4}-z-\tfrac{n}{c}\right)\left(\tfrac{23n}{200}+z-\tfrac{n}{c}\right) \geq \tfrac{n^2}{100}-z^2 .
\end{align*}
Since $|W_f \setminus Z_f|$ and $|W_{f+1}\setminus Z_{f+1}|$ are uniform random variables in $H$ (resp. $H'$) with $|H|,|H'|\geq n-\tfrac{n}{c}$, each value of pairs $h=|W_f\setminus Z_f|$ and $h' = |W_{f+1}\setminus Z_{f+1}|$ has probability at least $1/n^{2}$.
Therefore,
\begin{align*}
    \Prob{S_2(f) \mid \mathcal{E}(f,\Tilde{w}_f)} & \geq \sum_{z=1}^{\infty}\Prob{|Z_f| = z \mid \mathcal{E}(f,\Tilde{w}_f)} \cdot \frac{1}{n^2}\cdot \left(\frac{n^2}{100}-z^2\right).
\end{align*}
We remark that from the stochastic domination we have that $\Prob{|Z_f| \leq z \mid \mathcal{E}(f,\Tilde{w}_f)} \geq 1-\left(\tfrac{1}{c}\right)^{z}$,  and so, for large enough $c$, it holds
\begin{align*}
\Prob{S_2(f) \mid \mathcal{E}(f, \Tilde{w}_f)} &\geq \sum_{z = 1}^{\infty} \Prob{|Z_f| = z \mid \mathcal{E}(f,\Tilde{w}_f)}\cdot \frac{1}{n^2}\left(\frac{n^2}{100}-z^2\right)
    \\ & \sum_{z=1}^{n/100}\Prob{|Z_f| = z \mid \mathcal{E}(f,\Tilde{w}_f)} \left(\frac{1}{100}-\frac{1}{(100)^2}\right)
    \\ & \geq \Prob{|Z_f| \leq \tfrac{n}{100}\mid \mathcal{E}(f,\Tilde{w}_f)} \cdot \frac{1}{102}
     \geq \frac{1}{105}.
    \end{align*}
Therefore, from the inequality above and from \eqref{eq:bound_NF} we have that
\begin{align*}
    \Prob{s_\ell \in Z_f \mid \mathcal{E}(f,\Tilde{w}_f)} &\leq \frac{210}{n} \Prob{S_2(f) \mid \mathcal{E}(f,\Tilde{w}_f)},
\end{align*}
which proves \eqref{eq:inequality-successful} in the second case.

\medskip

\noindent {\textbf {STEP 5: Proof of the worst-case bound: conclusion.}}

\smallskip

We conclude now the proof by showing that $A(r_\ell,s_\ell) \leq \frac{215}{n}+{\rm e}^{-(s_\ell-t_\ell)}$. From \eqref{eq:def-a-requests} and the definition of $\mathcal{C}$,  we have that 
\begin{align*}
A(r_\ell,s_\ell) &=\Prob{r_\ell \in Q_{s_{\ell}} \mid \mathcal{C}}
\\ & \stackrel{(I)}{=} \sum_{f} \sum_{\Tilde{w}_f} \Prob{r_\ell \in Q_{s_\ell}, F= f, \mathcal{E}(f,\Tilde{w}_f) \mid \mathcal{C}} + \Prob{r_\ell \in Q_{s_{\ell}}, F = 0 \mid \mathcal{C}}
    \\ & \stackrel{(II)}{=} \sum_f \sum_{ \Tilde{w}_f}\Prob{s_\ell \in Z_f \mid \mathcal{E}(f,\Tilde{w}_f)}\Prob{\mathcal{E}(f,\Tilde{w}_f) \mid \mathcal{C}} + \left(\frac{1}{c}\right)^{s_\ell-t_\ell}
    \\ & \stackrel{(III)}{\leq} \sum_f \sum_{\Tilde{w}_f <s_\ell -\frac{101}{100}n}\Prob{s_\ell \in Z_f \mid \mathcal{E}(f,\Tilde{w}_f)}\Prob{\mathcal{E}(f,\Tilde{w}_f) \mid \mathcal{C}}
    \\ &+\tfrac{210}{n} \,  \sum_f \sum_{\Tilde{w}_f\geq s_\ell - \frac{101}{100}n}\Prob{S_1(f) \mid \mathcal{E}(f,\Tilde{w}_f)} \Prob{\mathcal{E}(f,\Tilde{w}_f) \mid \mathcal{C}}
    \\ & \notag
    + \tfrac{210}{n} \sum_f \sum_{\Tilde{w}_f \geq s_\ell - \frac{101}{100}n} \Prob{S_2(f) \mid \mathcal{E}(f,\Tilde{w}_f)} \Prob{\mathcal{E}(f,\Tilde{w}_f) \mid \mathcal{C}} +  \left(\frac{1}{c}\right)^{s_\ell-t_\ell}
    \\ & \stackrel{(IV)}{\leq} 
    \sum_f \sum_{\Tilde{w}_f <s_\ell- \frac{101}{100}n}{\rm e}^{-\frac{n}{100}}\Prob{\mathcal{E}(f,\Tilde{w}_f) \mid \mathcal{C}}
    \\&+\tfrac{210}{n}\sum_f \sum_{\Tilde{w}_f} \Prob{r_\ell \not \in Q_{s_\ell}, F  = f, |W_{F}\setminus Z_F| \geq \tfrac{7n}{8} \mid \mathcal{E}(f,\Tilde{w}_f)}\Prob{\mathcal{E}(f,\Tilde{w}_f) \mid \mathcal{C}}
    \\ &  + \tfrac{210}{n}\sum_f \sum_{\Tilde{w}_f} \Prob{r_\ell \not \in Q_{s_\ell},  F = f+1, |W_{F}\setminus Z_F| < \tfrac{7n}{8} \mid \mathcal{E}(f, \Tilde{w}_f)}\Prob{\mathcal{E}(f,\Tilde{w}_f) \mid \mathcal{C}} + \left(\frac{1}{c}\right)^{s_\ell-t_\ell}
    \\ & \stackrel{(V)}{\leq}  n^2 {\rm e}^{-\frac{n}{100}}+\tfrac{210}{n}\sum_f \sum_{\Tilde{w}_f}\Prob{r_\ell \not \in Q_{s_\ell}, F = f, |W_F \setminus Z_F| \geq \tfrac{7n}{8}, \mathcal{E}(f, \Tilde{w}_f) \mid \mathcal{C}}
    \\ & +  \tfrac{210}{n}\sum_f \sum_{\Tilde{w}_f}\Prob{r_\ell \not \in Q_{s_\ell}, F = f+1, |W_F \setminus Z_F|  < \tfrac{7n}{8}, \mathcal{E}(f, \Tilde{w}_f) \mid \mathcal{C}} +  \left(\frac{1}{c}\right)^{s_\ell-t_\ell}
    \\ & \stackrel{(VI)}{\leq} n^2 {\rm e}^{-n/100} + \tfrac{210}{n} \cdot \Prob{r_\ell \not \in Q_{s_\ell} \mid \mathcal{C}} + \left(\frac{1}{c}\right)^{s_\ell-t_\ell}\leq \tfrac{215}{n} + {\rm e}^{-(s_\ell-t_\ell)}.
\end{align*}

Where $(I)$ follows from the law of total probability, $(II)$ from the definition of $Z_f$ and since $\Prob{r_\ell \in Q_{s_\ell}, F = 0 \mid \mathcal{C}} \leq \left(\frac{1}{c}\right)^{s_\ell-t_\ell}$. The inequality $(III)$ follows from \eqref{eq:inequality-successful},  and $(IV)$ from \eqref{eq:bound-big-wf} and from the definition of $S_1(f)$ and $S_2(f)$ (in \eqref{eq:triplet-successful-cond-1} and \eqref{eq:triplet-successful-cond-2}). The inequality $(V)$ follows from an union bound over  all possible pairs of $f$ and $\Tilde{w}_f$, which are at most $n^2$.
The inequality  $(VI)$ follows from the fact that the events
\[\{F = f, |W_F\setminus Z_F| \geq \tfrac{7n}{8}, \mathcal{E}(f,\Tilde{w}_f)\}, \{ F = f+1, |W_F\setminus Z_F| < \tfrac{7n}{8}, \mathcal{E}(f, \Tilde{w}_f)\},\]
are disjoint and at the varying of $f$ and $\Tilde{w}_f$ are a (disjoint) partition of the event $\mathcal{C}$.
\end{proof}

\subsection{On the number of pending requests} \label{ssec:pending}
As we observed in the previous section,   the queue $Q_t$ (i.e. the set of all pending requests at round $t$) plays a crucial role in our probabilistic analysis. In particular, we will often  exploit the following upper bound on its size.



\begin{lemma} 
There exist constants $c$ and $d$ sufficiently large such that, for all $n$ large enough, the following holds.
For every $t  \geq 2n$, 
$$
\Prob{|Q_t| \leq 100(cd)^2 \log n}\geq 1-n^{-2}\,.
$$
\label{lem:size-queue}
\end{lemma}
To prove \cref{lem:size-queue}, we need the following preliminary lemma    showing that, if the size of the queue is larger than a suitable logarithmic threshold, then in the next round it decreases by a constant factor, w.h.p. 

\begin{lemma}
\label{lem:decreasing-queue}
There exist constants $c$ and $d$ sufficiently large such that, for all $n$ large enough, the following holds.
    For any $t \geq 2n$,   if $|Q_t| \geq 3 \cdot 32(cd)^2 \log n$, then 
    
    \begin{equation} \label{eq:queuedecr}
    \Prob{ |Q_{t+1}| \leq \tfrac{1}{2}|Q_t| \mid Q_t, G_{t-1} } \geq 1-n^{-3}.
\end{equation}       
\end{lemma}

\begin{proof}
Denote with $Y_t$ the number of requests in $Q_{t}$ accepted during round $t$, i.e. $Y_{t} = |Q_t \setminus Q_{t+1}|$ and let $W_t$ the new pending requests at time $t+1$, i.e., $W_t = |Q_{t+1}\setminus Q_t|$. Then, for any $t \geq 1$,
\begin{equation}
    \label{eq:increasing-queue}
    |Q_{t+1}| = |Q_t| + W_t - Y_t.
\end{equation}
Notice that, for any $t \geq 1$, we  deterministically have that $|W_t| \leq (c+1)d$: indeed,  each vertex has in-degree bounded by $c d$, and so a vertex leaving the graph can generate at most $cd$ new pending requests, while a vertex joining the graph always generate new $d$ pending requests. 

The distribution of the random variable $Y_t$  depends on the size of $Q_t$ and on the configuration of the graph at time $t$.
We next prove the following property of $Y_t$:  for any $t \geq 2n$,
\begin{equation}
  \Expc{Y_t \mid Q_t,G_{t-1}} \geq |Q_t| \left(1-\tfrac{4}{c}\right).  \label{eq:expectation-accepted-requests-final} 
\end{equation}
Denote with $S$ the set of vertices in $V_t$ with in-degree $\leq cd/ 2$: using  an argument similar to  that proving \cref{claim:bound-overloaded-nodes}, we have that $|V_t \setminus S| \leq 2n/c$. Let $Z_t$ be the set of  requests in $Q_t$ targeting   $S$, and colliding with less than $\tfrac{cd}{2}$ other requests. Notice that  the requests in $Z_t$ will thus be accepted at round $t$, and clearly $Y_t \geq |Z_t|$. Therefore, we can bound the conditional expectation by
\begin{equation}    \label{eq:expectation-accepted-requests}
\Expc{Y_t \mid Q_t,G_{t-1}} \geq \Expc{Z_t \mid Q_t,G_{t-1}} = \sum_{r=1}^{|Q_t|}\Prob{Z_t(r)=1 \mid Q_t,G_{t-1}},
\end{equation}
where $Z_t(r)$ is the binary random variable indicating whether the request $r$ is in the set $Z_t$. Hence, for each $r = 1,\dots, |Q_t|$, it holds that 
\begin{align}
   & \Prob{Z_t(r) = 1\mid Q_t,G_{t-1}} 
  \notag
   \\ & \geq \sum_{v \in S\setminus u} \Prob{\text{$r$ targets $v$ at round $t$}        \mid Q_t,G_{t-1}}  
   \Prob{
  \begin{tabular}{c}
    \text{$v$ is targeted by $< \tfrac{cd}{2}$ requests} \\
    \text{$ \neq r$ at round $t$}
  \end{tabular} \mid Q_t,G_{t-1}} \notag
    \\ & \geq \sum_{v \in S \setminus u} \tfrac{1}{n-1} \cdot \Prob{\Bin\left(|Q_t|,\tfrac{1}{n}\right) < \tfrac{cd}{2} \mid Q_t,G_{t-1}} , 
    \label{eq:bound-single-accepted}
\end{align}
where $u$ is the vertex making the  request $r$.
Since $|Q_t| \leq nd$, we have that
\[\Prob{\Bin\left(|Q_t|,\tfrac{1}{n}\right)\geq \tfrac{cd}{2} \mid Q_t,G_{t-1}} \leq \Prob{\Bin \left(nd,\tfrac{1}{n}\right)\geq \tfrac{cd}{2}} \leq {\rm e}^{-\left(\frac{c}{2}-1\right)^2 \frac{d}{3}} \leq  \tfrac{3}{d}\left(\tfrac{2}{c-2}\right)^2 \leq \tfrac{1}{c},\]
where the last inequality follows since $d \geq 3$ and $c \geq 16$. Hence, we can bound \eqref{eq:bound-single-accepted} with
\begin{align*}
    \Prob{Z_t(r) = 1 \mid Q_t,G_{t-1}} &\geq \sum_{u \in S \setminus v}\tfrac{1}{n-1} \left(1-\tfrac{1}{c}\right)
    \\ & \geq \tfrac{|S|-1}{n-1}\left(1-\tfrac{1}{c}\right)
    \\ & \geq \left(1-\tfrac{4}{c}\right),
\end{align*}
where we assumed $n \geq 3$. Hence, from the above  inequality   and from \cref{eq:expectation-accepted-requests}, we get  \cref{eq:expectation-accepted-requests-final}.

\smallskip
We have thus  bounded the expectation of $Y_t$, given $Q_t$ and the snapshot   $G_{t-1}$. Our next goal is  to prove that a concentration result on  $Y_t$. To this aim,  we  next  apply   the method of bounded differences (see \Cref{app:prob}).
Denote with $X_r$ the target vertex of the request $r$ at round $t$. We notice that, given $G_{t-1}$ and $Q_t$, $Y_t$ can be expressed as a function of $X_1,\dots, X_{|Q_t|}$. Indeed, once    the targets of the pending requests are fixed, we can   determine   which vertices will accept the requests, since we also have knowledge of $G_{t-1}$. Denote with $f$ the function (depending on $G_{t-1}$ and $Q_t$) such that
\[Y_t = f(X_1,\dots, X_{|Q_t|}).\]
Notice that $f$ satisfies the Lipschitz property (see \cref{def:lipschitz-property} in Appendix \ref{app:prob}) with coefficient $2cd$, and in particular
\[|f(x_1,\dots, x_r,\dots, x_{|Q_t|})- f(x_1,\dots, x_r',\dots, x_{|Q_t|})| \leq 2cd.\]
Indeed, denoting $x_r = v$ and $x_r' = v'$, the change of destination  of the request $r$ from $v$ to $v'$ can generate (in the worst-case)  the following two changes: 
\begin{enumerate}[(i)]
    \item     $v$ can accept at most $cd$ additional requests,
    \item      $v'$ can reject  at most $cd$ requests.
\end{enumerate}
 We thus get a maximum  total variation   $2cd$ for  $f$.
From \cref{thm:bounded-diff} and since $c \geq 16$,  we get 
\begin{align*}
\Prob{Y_t \leq \tfrac{5}{8}|Q_t| \mid Q_t,G_{t-1}} &\leq 
  \Prob{Y_t < \Expc{Y_t \mid Q_t,G_{t-1}} - \tfrac{1}{8}|Q_t| \mid Q_t, G_{t-1}}
\\ &  \leq {\rm e}^{-\frac{|Q_t|}{32(cd)^2}}
\\ & \leq {\rm e}^{-3 \log n} = n^{-3},
\end{align*}
where the last inequality follows from the hypothesis  $|Q_t| \geq 3 \cdot 32(cd)^2 \log n$.
Finally, from \eqref{eq:increasing-queue}, with probability at least $1-n^{-3}$,
\begin{align*}
  |Q_{t+1}| &\leq |Q_t| + (c+1) d - Y_t  
 \\ & \leq \tfrac{3}{8}|Q_t| + (c+1)d
 \\ & \leq \tfrac{1}{2}|Q_t| . 
\end{align*}
\end{proof}

Now we are ready to prove \cref{lem:size-queue}. 
\begin{proof}[Proof of \cref{lem:size-queue}]  
We first remark that, for any round $t \geq 2n$,  the following  two key facts (deterministically) hold:
\begin{enumerate}[(i)]
    \item \label{obs:case(1)-queue}  $|Q_{t-n}| \leq nd$;
    \item \label{obs:case(2)-queue} At each round $s \in \mathbb{N}$,   $|Q_{s+1}| \leq |Q_s| + (c+1)d$.
\end{enumerate}
For any round  $s \in \{t-n,\dots, t\}$, consider the event $\{$the queue $Q_s$ has size at most $3 \cdot 32(cd)^2 \log n$ or it halves from round $s$ to $s+1$$\}$, formally:
\[A_s = \{|Q_s| \leq 3 \cdot 32(cd)^2\log n\} \cup \{|Q_{s+1}| \leq \tfrac{1}{2}|Q_s|\},\]
 From \cref{lem:decreasing-queue} and from an union bound over the rounds $s=t-n,\dots, t$, it follows that 
\[\Prob{\cap_{s=t-n}^t A_s} \geq 1-n^{-2}.\]
We now proceed to show that 
\begin{equation}
  \cap_{s=t-n}^t A_s \subseteq \{|Q_t| \leq 100\log n \}.
  \label{eq:As-implies-small-queue}
\end{equation}
Let $\tau$ be the first round in which the queue has size less than $3 \cdot 32(cd)^2 \log n$,
\[\tau = \min\{s \geq t-n:|Q_t| \leq 3 \cdot 32(cd)^2\log n\}.\]
From \cref{obs:case(1)-queue}, the event $\cap_{s=t-n}^t A_s$ implies that $\tau \leq \log(nd)$. 
 Then, for every round $s \geq \tau$, it holds \cref{obs:case(2)-queue}: hence, every time $|Q_s| \leq 3 \cdot 32(cd)^2 \log n$, we can say that $|Q_{s+1}| \leq 3 \cdot 32(cd)^2 \log n + (c+1)d \leq 100 (cd)^2 \log n$, and every time that $|Q_{s+1}| \geq 3 \cdot 32(cd)^2 \log n$, the event $A_s$   implies that  $|Q_{s+2}|$ halves in the next round, thus
\[|Q_{s+2}| \leq \tfrac{1}{2}\left(3 \cdot 32(cd)^2 \log n +(c+1)d\right) \leq \tfrac{1}{2}\cdot 100(cd)^2 \log n \leq 3 \cdot 32(cd)^2 \log n.\] 
We have  thus  proved  \eqref{eq:As-implies-small-queue} which concludes the proof of the lemma.
\end{proof}




\subsection{On the number of pending rounds of a request}

The following lemma provides a bound on the overall number of rounds in which a fixed request $r$ is pending during all of its lifetime, namely on the quantity 
$$
P(r) = \sum_{t= t_r}^{t_r + n}\mathds{1}\left[r \in Q_t\right].
$$
\begin{lemma}
\label{lem:number-times-pending}
There exist constants $c$ and $d$ sufficiently large such that, for all $n$ large enough, the following holds.
For any $t \geq 2n$, any request $r$ in $V_t \times [d]$ verifies
\begin{align*}
    \Prob{P(r) \geq j} \leq 2{\rm e}^{-j/24}.
    \end{align*}
As a consequence
$$ \Expc{P(r)}= O(1)
$$
and in particular
\[\Prob{P(r) \geq 50\log n} \leq n^{-2}.\]
\end{lemma}

\begin{proof}
Consider the random variables $W_0,W_1,\dots$ and $Z_0, Z_1,\dots$ defined in \eqref{eq:def-Wi} and \eqref{eq:def-Zi} in \cref{lem:multiple-requests-destination}, with $r_\ell = r$ and $t_\ell = t_r$. Look also at \cref{fig:requests-r-ell} for a better understanding. As in \cref{lem:multiple-requests-destination}, we define $F$ as the last interval $W_F$ intersecting $\{t_r,\dots, t_r + n\}$.
Hence, we can write $P(r)$ as
\[P(r) = \sum_{f=0}^F |Z_f|.\]

If we define $Z_f = 0$ for any $f \geq F$, we then have that
\begin{align}
\notag
    \Prob{P(r) \geq j} &= \Prob{\sum_{f = 0}^F|Z_f|\geq j, F \geq j/4} + \Prob{\sum_{f = 0}^F|Z_f| \geq j, F < j/4}
    \\ & \leq  \Prob{F \geq \tfrac{j}{4}} + \Prob{\sum_{f = 0}^{j/4}|Z_f|\geq j}. \label{eq:bound-P(r)}
\end{align}

From \cref{claim:bound-overloaded-nodes}, we know that $|B_t|\leq \frac{n}{c}$ for every $t$, and hence the sum $\sum_f Z_f$ is stochastically bounded by the sum of i.i.d. geometric random variables $Y_f$ with parameter $\left(1-\frac{1}{c}\right)$. Therefore, 
\begin{align*}
    \Prob{\sum_{f=0}^{j/4}|Z_f| \geq j} &\leq \Prob{\sum_{f=0}^{j/4}Y_f \geq j}
    \\ & \stackrel{(I)}{\leq} \Prob{\mathrm{Bin}\left(j,1-\tfrac{1}{c}\right) \leq j/4}
    \\ & \leq {\rm e}^{-j/24},
\end{align*}
where $(I)$ follows from \cref{lem:geometric_distribution}, and the last inequality from Chernoff's Inequality (\cref{thm:chernoff}) and from the fact that $c \geq 2$.

In order to bound $\Prob{F \geq j/4}$, note that if $F \leq j/4$, then it holds that $\sum_{f=0}^{j/4}|W_f| < n$. 
Notice that, when $r$ is pending and targets a not full vertex $v$ at time $s$, that connection will remain active for $n-\age_s(v)$ steps, which determines the length of $|W_f\setminus Z_f|$ for appropriate $f$. Observe that, since $v$ is sampled uniformly at random in the not full vertices, and since the sampling are independent in each round, then the random variable $\sum_{f=0}^{j/4}|W_f\setminus Z_f|$ is stochastically \emph{lower} bounded by $\sum_{f=0}^{j/4}U_f$, where $U_f$ are uniform independent random variables in $\{1,\dots, n-\tfrac{n}{c}\}$. Hence,
\begin{align*}
    \Prob{F \geq j/4} &\leq \Prob{\sum_{f=0}^{j/4}|W_f| < n}
    \\ & \leq \Prob{\sum_{f=0}^{j/4}|W_f \setminus Z_f| < n}
    \\ & \leq \Prob{\sum_{f=0}^{j/4}U_f < n}
    \\ & \leq {\rm e}^{-j/8},
\end{align*}
where the last inequality follows from Hoeffding Bound (\cref{thm:hoeffding}).

In conclusion, considering \eqref{eq:bound-P(r)}, we showed that
\begin{align*}
    \Prob{P(r) \geq j} \leq 2{\rm e}^{-j/24}.
    \end{align*}
\end{proof}

\section{Expansion Properties} \label{sec:expans}

In this section, we will prove the main result of this paper that we re-state here in a more formal way.

\begin{theorem}
    \label{thm:main-expansion}
   Let  $n_0,c_0,d_0 \in \mathbb{N}$  and   $\alpha = \alpha(d)$ sufficiently large integers. Then, for  any $d \geq d_0$, $c \geq c_0$ and $n \geq n_0$, an integer $\beta = \beta(c,d)$ exists, such that 
    the snapshot $G_t = (V_t,E_t)$ generated by the $\BSDG(n,d,c)$ model with $t \geq 2n$ satisfy the following properties, w.h.p.
    \begin{enumerate}[(a)]
     \item For every  $S \subseteq V_t$ with $|S| \geq \beta \log n$ has conductance $\phi_t(S) \geq \alpha$;
     \label{item:main-expansion:bigsets}
        \item A subset $ H_t \subseteq V_t$ with $|H_t| =n - O(\log n)$ exists such that $G_t[H_t]$ is an $\alpha$-expander.
        \label{item:main-expansion:subset}
    \end{enumerate}
\end{theorem}


The proof of Claim~\ref{item:main-expansion:bigsets}  is given in the next two subsections:  the first one considers the vertex expansion of subsets of size in the range $\left[\beta \log n, \tfrac{n}{2000}\right]$, while the second one covers the remaining size range.\footnote{The factor $\tfrac{1}{2000}$ has been set in order to simplify some calculations: the optimization of this parameters is out of the scope of our analysis.}
In both cases, our analysis  will show a constant lower bound of $\varepsilon = \tfrac{1}{10}$ on the \textit{vertex expansion} of the considered vertex subsets. However, since the graph snapshots in 
$\BSDG(n,d,c)$ has  bounded maximum degree (i.e. $\leq (c+1)d)$, by definition of conductance (see \Cref{sec:prely}),  the latter will be at least  $\varepsilon ((c+1)d)^{-1} = \Omega(1)$. We recall that the \emph{vertex expansion} of the graph $G_t$ is defined as
\[h(G_t) = \min_{\substack{S \subseteq V_t:\\ |S| \leq \frac{n}{2}}}\frac{|\Gamma_t(S)|}{|S|}.\]

The proof of Claim~\ref{item:main-expansion:subset} of the main theorem above is provided in \Cref{ssec:subexp}, and it also consists of analyzing  the vertex-expansion of the considered subgraph.

\subsection{Expansion of small subsets}
\label{sec:small-expansion}

The goal of this section is to prove the following result.
\begin{lemma}[Expansion of small subsets]
There exist constants $c$ and $d$ sufficiently large such that, for all $n$ large enough, the following holds.
For any $t \geq 2n$ let $E_t$ be the event
    \[E_t = \left\{\min_{\substack{S\subseteq V_t\\ 2 \beta \log n\leq |S| \leq \tfrac{n}{2000}}} \frac{|\Gamma_t(S)|}{|S|} \geq \frac{1}{10}\right\} \]
    where $\beta = 100(cd)^2$. Then
    $$
    \Prob{E_t}\geq 1-n^{-2}.
    $$
\label{lem:small-expansion}
\end{lemma}

\begin{proof}


Firstly, observe that the complementary event of $E_t$ occurs if there exists a subset $T \subseteq V_t \setminus S$ such that $|T| = \lceil\frac{1}{10}|S|\rceil$ and $\Gamma_t(S) \subseteq S \cup T$, which implies that every request from vertices in $S$ has either destination in $S \cup T$, or it is pending. From now on, we will just suppose that  and $|S|/10$ is an integer, for simplicity.

Due to the dynamics of the graph and the bounded capacity of the vertices, any expansion result requires a large  number of accepted requests. More in detail,  we can ensure that, for $|S| \geq 2 \beta \log n$, most of its connection requests are accepted.  Indeed, consider the event
\[A =\{|Q_t| \leq \beta \log n\}, \]
then it holds that
\begin{align}
   \Prob{E_t^c} &\leq \Prob{ E_t^c \cap A} + \Prob{A^c} 
   \leq \Prob{E_t^c \cap A} + n^{-2}
   \label{eq:bound-small-subset-a}
\end{align}
where the last inequality follows from \cref{lem:size-queue}. 
From the previous remarks, and by a union bound on all possible choices of $S,T$, we can write
\begin{align}
    \Prob{E_t^c \cap A} \leq \sum_{\substack{S \subseteq V_t: \\ 2 \beta \log n \leq |S| \leq \frac{n}{2000}}} \sum_{\substack{T \subseteq V_t \setminus S: \\ |T| =  0.1|S|}} \Prob{\{\Gamma_t(S)\subseteq T\} \cap A},
\label{eq:thm:small-exp:1}
\end{align}
and we are left with estimating
$\Prob{ \{\Gamma_t(S) \subseteq T\} \cap A}$. Note that, if events $A$ and $\Gamma_t(S) \subseteq T$ hold and $|S| >2 \beta\log n$, there exists a subset of requests $R\subseteq S \times [d]$ with $|R| = d|S| - \beta \log n$ that are accepted with destination in $S \cup T$. Hence, 
\begin{align*}
\Prob{\{\Gamma_t(S)\subseteq T\} \cap A} &\leq \Prob{\exists R \subseteq S \times [d] \text{ s.t. } \cap_{r \in R} \{X_t(r) \in S \cup T\}}   
\\ & \leq \sum_{\substack{R\subseteq S \times [d]: \\ |R|= d|S| - \beta \log n}} \Prob{\cap_{r \in R}\{X_t(r) \in S \cup T\}}
\\ & \leq  \sum_{\substack{R\subseteq S \times [d]: \\ |R|= d|S| - \beta \log n}} \left(\frac{220 \cdot \left(1+\frac{1}{10}\right)|S|}{n-1}\right)^{|R|},
\end{align*}
where the last inequality follows from \cref{lem:multiple-requests-destination}.
Going back to \eqref{eq:thm:small-exp:1}, we have
\begin{align*}
    \Prob{E_t^c \cap A} & \leq \sum_{\substack{S \subseteq V_t: \\ 2 \beta \log n \leq |S| \leq \frac{n}{2000}}} \sum_{\substack{T \subseteq V_t \setminus S: \\ |T| = 0.1|S|}} 
    \sum_{\substack{R\subseteq S \times [d]: \\ |R|= d|S| - \beta \log n}} \left(\frac{242|S|}{n-1}\right)^{|R|}
    \\ 
    &= \sum_{s = 2\beta \log n}^{n/2000} \binom{n}{s}\binom{n-s}{\frac{1}{10}s}\binom{ds}{ds-\beta\log n} \left(\frac{242s}{n-1}\right)^{ds-\beta \log n}
    \\ & \stackrel{(I)}{\leq} \sum_{s = 2\beta \log n}^{n/2000} \left(\frac{30n}{s}\right)^{\frac{11}{10}s} \left(\frac{1500s}{n-1}\right)^{ds-\beta \log n}
    \\ & \stackrel{(II)}{\leq }\sum_{s = 2\beta \log n}^{n/2000} \left(\frac{30n}{s}\right)^{\frac{11}{10}s}\left(\frac{1500s}{n-1}\right)^{s(d-1)}
    \\ & \stackrel{(III)}{\leq} \sum_{s=2\beta \log n}^{n/2000} \left(\frac{1500s}{n-1}\right)^{s(d-43)}
    \\ & \leq \sum_{s=2\beta \log n}^{n/2000}\left(\frac{1}{2}\right)^{2\beta \log n(d-43)} \leq n^{-2},
\end{align*}
where in $(I)$ we used the fact that, for any $k \leq n$, $\binom{n}{k}\leq \left(\frac{n {\rm e}}{k}\right)^k$, in $(II)$ we used the fact that $ds - \beta \log n \geq s(d-1)$, and in $(III)$ we used the fact that $s \leq \tfrac{n}{2000}$. The last inequality follows by considering $d$ large enough. The lemma follows then from \eqref{eq:bound-small-subset-a}.
\end{proof}

\subsection{Expansion of large subsets}

\label{sec:large-expansion}

The goal of this section is to prove the following result.

\begin{lemma}[Expansion of big subsets]
There exist constants $c$ and $d$ sufficiently large such that, for all $n$ large enough, the following holds.
    For any $t \geq 2n$ let $E_t$ be the event
    \[E_t = \left\{\min_{\substack{S\subseteq V_t\\\frac{n}{2000}\leq |S| \leq \tfrac{n}{2}}} \frac{|\Gamma_t(S)|}{|S|} \geq \frac{1}{10}\right\}\,. \]
    Then
    $$
    \Prob{E_t}\geq 1-{\rm e}^{-n}\,.
    $$
    \label{lem:big-expansion}
\end{lemma}
\begin{proof}
Fix any subsets $S \subseteq V_t$ of size $\frac{n}{2000} \leq |S| \leq \frac{n}{2}$ and $T \subseteq V_t \setminus S$ such that $|T| = \lceil\frac{1}{10}|S|\rceil$ (from now on we will just suppose that  $n/2000$ and $|S|/10$ are integers, for simplicity). Taking $P = S \cup T$ and $P^c = V_t \setminus P$, we have that
\begin{align}\label{gatto}
        \Prob{\Gamma_t(S) \subseteq T} 
        = \Prob{\cap _{r\in S\times [d]}\,\{X_t(r)\not\in P^c\}}\,.
\end{align}
We note that for each $r\in S\times [d]$ it holds
\begin{align}\label{gatto1}
    \{X_t(r) \not\in P^c\} \subseteq F_{r}(P^c)
\end{align}
where, calling $t_r\leq t$ the round when the request $r$ joined the graph, for any $A \subseteq V_t$
    \[
        F_{r}(A) = \{r \text{ did not establish a connection with a vertex in $A$ when it joined the graph}\}\,.
    \]
Indeed, if request $r$ established a connection with some vertex of $P^c$ when it entered the graph at time $t_r$, then it would still be connected to $P^c$ at time $t\geq t_r$. 
Note that it is possible that not all vertices of $P^c$ were already in the graph at time $t_r$. 

For every vertex $a \in S$, consider now 
\begin{equation*}
    \mathcal{O}_a = \{b \in P^c \mid \age_t(a) < \age_t(b)\}\,
\end{equation*}
the subset of vertices in $P^c\subseteq V_t$ that were in the graph when $a$ joined it. Clearly, $F_{r}(P^c)=F_{r}(\mathcal{O}_{a(r)})$ if $r$ is a request from vertex $a(r)$. In the rest of the proof we will abbreviate $a(r)$ as $a$.  
Then, from \eqref{gatto} and \eqref{gatto1} we have
\begin{align}\label{gatto2}
    \Prob{\Gamma_t(S) \subseteq T}
    \leq \Prob{\cap_{r \in S\times [d]} F_{r}(\mathcal{O}_a)}\,.
\end{align}

Let $k =|S|$ and $\{a_1,\dots,a_k\}$ be an age-based ordering of the vertices in $S$ from the oldest to the youngest, so that $t_1 < \dots < t_k$. We will analyze the r.h.s.~of \eqref{gatto2} by subsequentially conditioning on the events involving older vertices. We start by writing 
\begin{align}
    \Prob{\cap_{r \in S\times [d]} F_{r}(\mathcal{O}_a)} = \Prob{\cap_{j=1}^d F_{(a_k,j)}(\mathcal{O}_k) \big| \cap_{i=1}^{k-1} \cap_{j=1}^d F_{(a_i,j)}(\mathcal{O}_i)}\Prob{\cap_{i=1}^{k-1} \cap_{j=1}^d F_{(a_i,j)}(\mathcal{O}_i)}
    \label{large-exp:4}
\end{align}
where we abbreviated  $\mathcal{O}_{a_i}$ as $\mathcal{O}_i$ to ease the notation. Let us focus on the conditional probability in the last expression.
Recall that any fixed $r \in \{a_k\}\times[d]$ may fail to establish a connection with $\mathcal{O}_k$ at time $t_r$ for two reasons: either because it targets a vertex outside of $\mathcal{O}_k$, or because it receives a rejection from the target vertex in $\mathcal{O}_k$. The first event occurs with probability \[\frac{n-1-|\mathcal{O}_k|}{n-1}\] since the targets are chosen uniformly at random independently from the past. The second event happens if the targeted vertex is full at time $t_r$, or if the vertex targeted by $r$ is also targeted by too many other requests in $Q_{t_r}$. As we are interested in the rejection of the $d$ requests $\{(a_k,j), \; j\in [d]\}$, by the principle of deferred decision we can assume that all $r' \in Q_{t_r}\setminus\{(a_k,j), \; j\in [d]\}$ are sent before $\{(a_k,j), \; j\in [d]\}$. Now, if any $r \in \{(a_k,j), \; j\in [d]\}$ targets a vertex that has an in-degree of at most $(c-1)d$ \emph{after all other requests in the queue are sent}, the attempt will certainly be accepted, independently from the other $r' \in \{(a_k,i), \; i\in [d]\} \setminus \{r\}$ and from what happened in the past. Therefore, if we call $\tilde B_{t_k}$ the set of vertices with load at least $(c-1)d$, at time $t_k$ and \emph{after all other  requests in the queue are sent}, the probability of $r$ being rejected is at most
\[\frac{|\mathcal{O}_k \cap \tilde B_{t_k}|}{n-1}\,.\]
Thus, we can conclude that
\begin{align*}
    \Prob{\cap_{j=1}^d F_{(a_k,j)}(\mathcal{O}_k) \mid \cap_{i=1}^{k-1} \cap_{j=1}^d F_{(a_i,j)}(\mathcal{O}_i)} &\leq 
    \left(\frac{n-1-|\mathcal{O}_k|}{n-1} + \frac{|\mathcal{O}_k \cap \tilde B_{t_k}|}{n-1}\right)^d 
    \\
    &=\left(1-\frac{|\mathcal{O}_k \cap \tilde B_{t_k}^c|}{n-1}\right)^d.
\end{align*}
The same argument can be iteratively applied to $\Prob{\cap_{i=1}^k \cap_{j=1}^d F_{(a_i,j)}(\mathcal{O}_i)}$, isolating $d$ requests per iteration, and it leads to
\begin{align}
    \Prob{\cap_{r \in S\times [d]} F_r(\mathcal{O}_a)} &\leq \prod_{i=1}^k \left(1-\frac{|\mathcal{O}_i \cap \tilde B_{t_i}^c|}{n-1}\right)^d 
    \overset{(I)}{\leq} 
    \exp{\left(-\frac{d}{n-1}\sum_{i = 1}^k\left|\mathcal{O}_i \cap \tilde B_{t_i}^c\right|\right)} \label{large-exp:5}
\end{align}
where inequality $(I)$ follows since $1+x \leq {\rm e}^x$.

Now, if we look at the set of possible pairs $(a,b) \in S \times P^c$, two cases may arise:
    \begin{enumerate}[(i)]
        \item $ \big|\{(a,b) \in S \times P^c \mid \age_t(a) < \age_t(b)\}\big| \geq  \frac{|S| \cdot |P^c|}{2}$, \label{case:older-nodes}
        \item $ \big|\{(a,b) \in S \times P^c \mid \age_t(a) > \age_t(b)\}\big| \geq  \frac{|S| \cdot |P^c|}{2}$. \label{case:younger-nodes}
    \end{enumerate}

If case \cref{case:older-nodes} holds, then
\begin{align*}
    \sum_{a \in S} |\mathcal{O}_a| &= \sum_{a \in S} |\mathcal{O}_a\cap \tilde B_{t_a}| + |\mathcal{O}_a \cap \tilde B_{t_a}^c| \geq \frac{|S| \cdot |P^c|}{2}
\end{align*}
which implies that
\begin{align*}
    \frac{d}{n-1}\sum_{a \in S} |\mathcal{O}_a \cap \tilde B_{t_a}^c| \geq \frac{d}{n-1}\left(\frac{|S| \cdot |P^c|}{2} - \sum_{a \in S} |\mathcal{O}_a \cap \tilde B_{t_a}|\right).
\end{align*}
Using the same argument of \cref{claim:bound-overloaded-nodes}, it can be shown that $|\tilde B_{t_a}| \leq \frac{n}{c-1}$, yielding
\begin{align}
    \frac{d}{n-1}\sum_{a \in S} |\mathcal{O}_a \cap \tilde B_{t_a}^c| \geq \frac{d}{n-1}\left(\frac{|S|\cdot|P^c|}{2} - \frac{n}{c-1}|S|\right) \overset{(I)}{\geq} \frac{d}{7}|S|
    \label{large-exp:6}
\end{align}
where in $(I)$ we used that $|S| \leq \frac{n}{2}$, that $|P^c| = n- \frac{11}{10}|S|$  and that we can take for example $c \geq 16$ as in Lemma \ref{lem:small-expansion}. By plugging \eqref{large-exp:6} in \eqref{large-exp:5}, we obtain
\begin{equation}
\Prob{\Gamma_t(S) \subseteq T} \leq \exp{\left(-\frac{d}{n-1}\sum_{i = 1}^k\left|\mathcal{O}_i \cap \tilde B_{t_a}^c\right|\right)} \leq \exp{\left(-\frac{d}{7}|S|\right)}.
\end{equation}

Then, analogously to what has been done in the proof of \cref{lem:small-expansion}, a union bound on all possible choices of $S,T$ leads us to
\begin{align*}
    \Prob{E_t}
    &\leq \sum_{s=n/2000}^{n/2} \binom{n}{s}\binom{n-s}{\frac{1}{10}s}{\rm e}^{-\frac{d}{7}s} \notag
    \\
    &\overset{(I)}{\leq} \sum_{s = n/2000}^{n/2} \left(\frac{n {\rm e}}{\frac{1}{10}s}\right)^{\frac{11}{10}s}{\rm e}^{-\frac{d}{7}s}
    \\
    &\overset{(II)}{\leq} \sum_{s = n/2000}^{n/2} {\rm e}^{\big(13-\frac{d}{7}\big)s}
    \\ & \leq \frac{n}{2}{\rm e}^{(13-\frac{d}{7})\frac{n}{2000}},
    \end{align*}
    where $(I)$ is since $ \binom{n}{k} \leq \big(\frac{n {\rm e}}{k}\big)^k$, while in $(II)$ we used that $s \geq \frac{n}{2000}$. 
    By taking $d$ sufficiently large, one obtains $\Prob{E_t}\leq {\rm e}^{-n}$.
    The proof can be completed when case \cref{case:younger-nodes} holds with the same argument, by considering the requests sent from $P^c$ to $S$. 
\end{proof}

\subsection{On the existence of an expander subgraph}
\label{ssec:subexp}

This subsection is devoted to the proof of the following lemma, which immediately implies Claim \ref{item:main-expansion:subset} of \cref{thm:main-expansion}.

\begin{lemma} \label{lem:expsubgraph}
There exist constants $c$ and $d$ sufficiently large such that, for all $n$ large enough, the following holds.
    A constant $\beta = \beta(c,d)>0$ exists such that the snapshot $G_t$ of $\BSDG(n,d,c)$ for any  $t \geq 2n$ verifies the following property. A subset $H_t \subseteq V_t$ with $|H_t| \geq n-\beta \log n$ exists such that the induced subgraph $G_t[H_t]$ has vertex expansion at least $\frac{1}{20}$, w.h.p.
\end{lemma}

\begin{proof}
    Fix $\beta = 100(cd)^2$.
    To prove the lemma, we   show that the following event holds w.h.p.
    \[E = \left\{\exists H_t \text{ with $|H_t|\geq n-\beta \log n$ s.t. $\min_{\substack{S\subseteq H_t: \\ |S| \leq n/2}} \frac{|\Gamma_t(S) \cap H_t|}{|S|} \geq \frac{1}{20}$}\right\}.\]

    Notice that, in \cref{lem:small-expansion} and \cref{lem:big-expansion}, we proved that all the sets $S\subseteq V_t$ with size at least $\beta \log n$ have vertex expansion at least $\frac{1}{10}$ w.h.p. 
    Therefore, to show that $E$ holds w.h.p., we need to prove that there exists a subset $H_t \subseteq V_t$ such that, all the sets $S \subseteq H_t$ with $S \leq 20\beta \log n$ have also vertex expansion at least $\frac{1}{20}$.
    Indeed, the fact that the subsets $|S| \geq \beta \log n$ have expansion at least $1/10$ implies directly that the event $E$ holds for such subsets, since, for each $S\subseteq H_t$ such that $|S| \geq 20 \beta \log n$, we have
    \[\frac{|\Gamma_t(S) \cap H_t|}{|S|} \geq \frac{|\Gamma_t(S)|}{|S|} - \frac{\beta \log n}{|S|}\geq \frac{1}{10}-\frac{1}{20} = \frac{1}{20}.\]
    
    In particular, we will take $H_t$ as all the set of nodes without pending requests at time $t$.
    More formally, we have that $E_1 \cap E_2 \cap E_3 \subseteq E$, where $E_1$ and $E_2$ are the events defined in the \cref{lem:small-expansion} and \cref{lem:big-expansion}, and
    \[E_3 = \left\{\exists H_t \text{ with $|H_t| \geq n - \beta \log n$ s.t. $\min_{\substack{S\subseteq H_t:\\ |S| \leq 20\beta \log n}} \frac{|\Gamma_t(S) \cap H_t|}{|S|} \geq \frac{1}{20}$}\right\}.\]
    From \cref{lem:small-expansion} and \cref{lem:big-expansion}, we have that $\Prob{E_1^c \cup E_2^c} \leq 4n^{-2}$. In what follows, we will show that $\Prob{E_3^c}\leq 2n^{-2}$.
    
   Notice that, from \cref{lem:size-queue}, we have that, if $A = \{|Q_t| \leq \beta \log n\}$, it holds that
   \begin{align}
       \Prob{E_3^c}  \leq \Prob{E_3^c \cap A} + \Prob{A^c} \leq \Prob{E_3^c\cap A} + n^{-2}.       
       \label{eq:bound_F}
   \end{align}
   If we define $\Tilde{Q}_t$ as the set of nodes $v \in V_t$ with at least one pending request, we have that $|\Tilde{Q}_t| \leq |Q_t|$ and that (taking $H_t = V_t \setminus \Tilde{Q}_t$)
   \begin{align*}
  E_3^c \cap A &\subseteq \left\{\text{$\exists S \subseteq V_t\setminus \Tilde{Q}_t$ s.t. $|S|\leq 20\beta \log n$, $|\Gamma_t(S) \setminus \Tilde{Q}_t|\leq \tfrac{1}{10}|S|$}\right\}
  \\ & \subseteq \{\text{$\exists S\subseteq V_t\setminus \Tilde{Q}_t$, $\exists T \subseteq V_t \setminus S$ s.t. $|S| \leq 20\beta \log n$, $|T| = \tfrac{1}{10}|S|$, $\Gamma_t(S) \subseteq T\cup \Tilde{Q}_t $}\}.
   \end{align*}
   Therefore, it holds that
   \begin{align*}
       \Prob{E_3^c \cap A} &\leq \sum_{\substack{S \subseteq V_t: \\ |S| \leq 20\beta \log n}}\sum_{\substack{T \subseteq V_t \setminus S:\\ |T| = \frac{1}{10} s}}\Prob{\Gamma_t(S) \subseteq T\cup \Tilde{Q}_t, S \cap \Tilde{Q}_t = \emptyset}
       \\ & = \sum_{\substack{S \subseteq V_t: \\ |S| \leq 20 \beta \log n}}\sum_{\substack{T \subseteq V_t \setminus S:\\ |T| = \frac{1}{10}s}}\Prob{\cap_{r \in S \times [d]}\{X_t(r)\in S \cup T \cup \Tilde{Q}_t\}}
       \\ & \stackrel{(I)}{\leq } \sum_{\substack{S \subseteq V_t: \\ |S| \leq 20\beta \log n}}\sum_{\substack{T \subseteq V_t \setminus S:\\ |T| = \frac{1}{10} s}} \left(\frac{220(|S| +\frac{1}{10}|S|+|\Tilde{Q}_t|)}{n}\right)^{d|S|}
\\ & \stackrel{(II)}{\leq} \sum_{s=1}^{20\beta \log n}\binom{n}{s}\binom{n-s}{\frac{1}{10}s} \left(\frac{242s+220\beta \log n}{n-1}\right)^{ds}
\\ & \stackrel{(III)}{\leq} \sum_{s=1}^{20\beta \log n}\left(\frac{30n}{s}\right)^{\frac{11}{10}s}\left(\frac{242s+220\beta \log n}{n-1}\right)^{ds}
\\ & \stackrel{(IV)}{\leq} \sum_{s=1}^{\beta \log n} \left(\frac{5060\beta \log n}{n-1}\right)^{(d-2)s}
\\ & \leq 2 \left(\frac{5060 \beta \log n}{n}\right)^{d-2}  \\ &\stackrel{(V)}{\leq}  n^{-2},
   \end{align*}
   where $(I)$ follows from \cref{lem:multiple-requests-destination}, $(II)$ from the fact that we are looking at $E_3^c \cap A$, hence $|\Tilde{Q}_t|\leq \beta \log n$, $(III)$ from the fact that, for any $k \leq n$, $\binom{n}{k}\leq \left(\frac{n {\rm e}}{k}\right)^k$, $(IV)$ from the fact that $s \leq 20\beta \log n$, and $(V)$ for $d$ large enough.

From \eqref{eq:bound_F}, since $\Prob{E_1^c \cup E_2^c} \leq 4n^{-2}$, and since $E_1\cap E_2 \cap E_3 \subseteq E$, it follows that 
\[\Prob{E^c} \leq \Prob{E_1^c \cup E_2^c} + \Prob{E_3^c} \leq 6n^{-2},\]
proving the lemma.
\end{proof}

\section{On the Convergence Time of \push\ and \pull }
\label{sec:pushpull}

\subsection{Rumor spreading on the \texorpdfstring{$\BSDG$}{BSG}\ model}
We     shortly recall   how    \push\   and   \pull\  \cite{demers1987epidemic} can be defined  on the $\BSDG$\ model. Such simple, local mechanisms are used to perform efficient broadcast operations over communication networks.

Given a connected  graph $G=(V,E)$ and a \textit{source} vertex $s \in V$, the goal is to inform all vertices about a piece of information that only $s$ initially knows. The synchronous, uniform \push\ protocol works as follows. At round $t=0$, the source selects one neighbor $v$ uniformly at random  and sends the message to it: we say that $v$ is \textit{informed} at (the end of) round $t$. Then, at every round $t \geq 1$, each informed vertex selects one random neighbors and sends the message to it. 
In the \pull\ protocol, each  node $u$,  which is  still not informed at (the beginning of) round $t$, selects one random neighbor $v$ and, if $v$ is informed, then $u$ pulls the source message from $v$ and gets informed. The \push -\pull~protocol is defined by considering both the \push\ and  \pull\ actions performed by each vertex, at every round.

In order to combine 
of the protocols  described above  with the process generated by the $\BSDG$\ model, we organize each  synchronous round $t \geq 1$  in two consecutive \textit{phases}. In the first, \textit{topology}  phase, all the actions of the $\BSDG$\ process described in \Cref{def:streaming-node-churn} and \Cref{def:edge-process} take places: this generates the snapshot $G_t$. Then, in the second \textit{rumor-spreading} phase, the local rule  of   \pull\ and/or    \push\   are applied by every vertex  in $V_t$ in parallel on $G_t$.

The aim of this section is to show that, in  the  $\BSDG$\ model,  such protocols completes  the   broadcast operation, starting from a new   source vertex,   within $O(\log{n})$ rounds, w.h.p. 

\begin{theorem} \label{thm:gossip} 
There exist constants $c$ and $d$ sufficiently large such that, for all $n$ large enough, the following holds. Let  $s$ be a \textit{source} node joining the  $\BSDG(n,d,c)$  dynamic graph at  some round $t_s \geq 2n$.   Then, after $T= O(\log n)$ rounds, the \push~or the \pull~protocol inform $n-O(\log n)$ vertices in $G_{T+t_s}$, w.h.p.
\end{theorem}

\subsection{Proof of \texorpdfstring{\cref{thm:gossip}}{Theorem 15}}

  \paragraph{Rumor spreading on static graphs: Previous results. }
 Our proof makes use of      the following important result  and its proof argument (see Theorem 12 in  \cite{CGLP18}) that bounds the completion time of rumor spreading protocols over static graphs of bounded degree. Below,  we  recall its statement and provide a short overview of its proof argument.

  \begin{theorem} [\cite{CGLP18}] \label{thm:rumor-conductance}
       Let $G=(V,E)$ be a connected $n$-vertices graph with conductance $\phi$ and such that, for any edge $\{u,v\} \in E$, $\grad(u)/\grad(v) = \Theta(1)$. Then, $O(\log n / \phi )$ rounds of \push\ or \pull\  suffice to spread to all nodes of $G$ a message originated at an arbitrary source node, w.h.p.
  \end{theorem}

\begin{proof}[Proof (outline)]
    Let us consider an almost-regular graph $G=(V,E)$ with constant conductance $\phi = \Theta(1)$.
    Let $I_t \subseteq V$ the set of informed nodes at round $t$ and assume that $|I_t| \leq n/2$. We first notice that, since $G$ is an almost-regular  $\Theta(1)$-expander, the size of the outer boundary of $I_t$ is such that $|\bord(I_t) | \geq \gamma |I_t|$, for some constant $\gamma >0$. Then, at every round $t' \geq t$, the \pull\ or the \push\ protocol let every node $v \in \bord(I_t)$ to have  constant probability to get informed. This implies that the expected number of informed nodes at round $t+1$  will be at least $(1+\Theta(1)) |I_t|$.
  By applying suitable concentration arguments, this fact is then used to show that, within $O(\log n)$ rounds, the number of informed nodes is at least $n/2$, w.h.p. Once the spreading process reaches at least $n/2$ informed nodes, the  analysis proceeds in a similar way by looking at the set  non-informed nodes at round $t$ and show that this quantity decreases at exponential rate, w.h.p.
  \end{proof}

  \paragraph{The analysis on   the  $\BSDG$ model.}
  In order to apply the above proof argument  on  the $\BSDG$ model we need to cope with two main technical issues.

  Our \Cref{lem:small-expansion} shows that, at any round $t \geq 2n $,   each  subset $S \subseteq V_t$ of the snapshot $G_t=(V_t,E_t)$ with   $|S| \geq \beta \log n$ for some constant $\beta >0$, has conductance $\phi_t(S) = \Omega(1)$, w.h.p. Then, in order to apply the proof argument of \Cref{thm:rumor-conductance}, we need to show that there is an initial phase of the rumor-spreading process, called \textit{bootstrap},  that is able to inform at least $\beta \log n$ vertices, w.h.p.
   Indeed, after this  bootstrap, we can apply the same   argument of  the proof of \Cref{thm:rumor-conductance}   assuming that there is an informed subset of logarithmic size.  The analysis of the bootstrap will be discussed later in this section.

   The second technical issue is caused by the presence of a set of \emph{old}  nodes (defined later in this section) that, during the information process, can leave the graph and create edge deletions and regenerations.
   However, once the bootstrap is completed, the subset of informed nodes reaches a logarithmic size which is large enough  to dominates the impact of all possible edge deletions that can take place for a time window of logarithmic size even in an adversarially fashion: this time window is exactly what the rumor spreading process needs to complete the broadcast task.  As we used in  several previous steps of our analysis,  this limited impact  is essentially due to the fact that the maximum vertex degree of the graph snapshots is always bounded by the constant quantity $(c+1)d$ and thus, at every round, only this number of edges can be deleted.

  \paragraph{The Bootstrap.} 

   Recall that $s$ is the source node joining the dynamic graph in round $t_s \geq 2n$.
Let $\OLD$ be the nodes in $V_{t_s}$ having age larger than $n- \log^2 n$.
Our   goal   is to prove the following.


\begin{lemma}
    Let $\beta> 0$. Then, 
   within $T'=O(\log n)$ rounds
    after the informed source $s$ joined the graph at time $t_{s}$,
    there are $\beta \log n$ informed nodes whose age is at most $n-\log^{2}n$, w.h.p.
    \label{lem:bootstrap}
\end{lemma}

\begin{proof}
From   \Cref{lem:expsubgraph}, at time $t_s$ when the source enters the graph, there exists a connected\footnote{ $G_{t_s}[H_{t_s}]$ is a vertex expander but in this proof we only use the fact that it is a large connected subgraph.} graph $G_{t_s}[H_{t_s}]$ with vertex subset  $H_{t_s}$ of size  $n-O(\log n)$. Consider now the connected components $\{C_i\}_{i\in I}$, obtained by removing from $H_{t_s}$ the set $\OLD$ of nodes that will die within the next $\Theta(\log^2 n)$ rounds. 
How many nodes in $\{C_i\}_{i\in I}$ belong to a connected component of size smaller than $\beta \log n$? Since $|I|\leq |\OLD| = \Theta( \log^2 n)$, there are at most $\beta \log n \cdot |I| = O(\log^3 n )$ such nodes.  
We now proceed with defining the following events:  Let $D$ be the event  ``$s$ is not connected after $2 \log n$ rounds'', $B$ be the event ``$s$ targets a node which is either in $\OLD$ or in a small component during some round  $t_s,...,t_s+2 \log n$'' and $C$ be the event  ``$s$ gets connected to a node in a large connected component that will remain connected for at least $\Theta( \log^2 n)$ rounds''.
 Now notice that, since $C^c \subseteq B \cup D$
\begin{equation} \label{eq:DBC}
    \Prob{C} = 1- \Prob{C^c} \geq 1-  \Prob{B \cup D} \geq 1- (\Prob{B} +\Prob{ D}).
\end{equation}

We then know that $\Prob{D} \leq O(n^{-2})$ by a standard concentration argument on the geometric probability of success (we again use \cref{claim:bound-overloaded-nodes} that says that, at every round, the number of full vertices is at most $n/c$). 
Observe also that, using   a union bound over the observed time window,  $\Prob{B}  = O( \polylog(n) /n)$, since at each round the probability that the request targets a node in $\OLD $ or in a small component (regardless of whether it is relaunched or not) is $O(\polylog (n)/n)$.

From the above facts and \Cref{eq:DBC}  we get that, w.h.p., the source $s$ will belong to a  subgraph of size at least $\beta \log n$
that will remain connected for at least $\Theta(\log^2 n)$ rounds. 
\end{proof}

Finally, thanks to \cref{lem:bootstrap}, we can apply the expansion argument we described in the proof sketch of \Cref{thm:rumor-conductance}  to the sets with size $\geq \beta \log n$ and get that, w.h.p., after $O(\log n)$ rounds, at least  $n - O( \log n)$ vertices in the graph will be informed (\cref{lem:small-expansion} using \cref{lem:big-expansion}).

\begin{remark} \label{rem:gossip}
Our analysis above proving \cref{thm:gossip} easily implies a further stabilizing property of the rumor spreading protocols on the $\BSDG$  model. In particular, after the  source joins the graph at round $t_s$, for a time window of a polynomial length,   every new vertex will get informed within $O(\log n)$ rounds w.h.p.
\end{remark}

\section{Further Motivations and  Related Work} \label{sec:related}


The  graph process we consider in this paper is natural and, as remarked in the introduction,     has the main merit of including crucial aspects of  the way some unstructured peer-to-peer networks maintain a
well-connected topology:    vertices joining and leaving the network,
 bounded degree and almost fully-decentralized network formation. For example,  full-vertices of the Bitcoin network~\cite{nakamoto2008bitcoin} running the Bitcoin Core implementation   rely on DNS seeds to allow full-vertices to   find active neighbors.   This allows them to pick new neighbors essentially at random among all vertices of the
network~\cite{bitnodes}.\footnote{In our model, this service is implemented by the link manager.} Notice also that the real topology of the Bitcoin network is hidden by the network formation protocol and discovering the real
network structure has been recently an active subject of
investigations~\cite{delgado2019txprobe,neudecker2016timing}.

Our analysis of the dynamic graph model $\BSDG$ focuses on two key aspects: expansion and the speed of information spreading. Beside having a theoretical interest, both of them play a crucial role for the resilience and the efficiency (in particular for the \textit{network delay})  of the unstructured 
peer-to-peer networks we discussed above: see \cite{albrecht2024larger,cruciani2022brief,cruciani2023dynamic}, for a deeper discussion of this issue. 


A basic way to classify dynamic graphs relies on whether the set of vertices stays the same or changes over time. If the vertex set is fixed, the graph is called an edge-dynamic graph, where only the edges change over time. Several formal models for this type of graph have been proposed and studied in depth in previous research \cite{CMMPS08,CMPS11,KLO10,KO11,M16}.
Conversely, the case in which the vertex set evolves over time has received less attention. This type of graph, usually described as a sequence of graphs $G_t = (V_t, E_t)$, for $t \geq  0$, is known as a dynamic network with churn \cite{augustine2016distributed}. In this setting, both vertex arrivals and departures (affecting $V_t$) and edge updates (affecting $E_t$) are governed by specific rules. The number of vertices that may join or leave the network in each time step is called the churn rate. For brevity, we will only review analytical results on dynamic networks with churn that are directly related to the models studied in this paper. 

As discussed in the introduction,  \cite{becchetti2023expansion} analyzes an unbounded-degree version of \ALG{}\ over both the streaming node-churn model and the \textit{continuous Poisson} one \cite{pandurangan2003building}: in the latter,    the number of births within each time 
unit follows a Poisson distribution with mean $\lambda$, and where the lifetime of each node is 
independently distributed as an exponential distribution with parameter  $\mu$, so that the average lifetime of a node is $1/\mu$ and the 
average number of nodes in the network at any given time is 
$\lambda/\mu$.   While this model is more realistic than the streaming one we consider in this work, we remark that in \cite{becchetti2023expansion} all expansion properties proved  in one model do hold in the other one as well, thus giving evidence of the robustness of the streaiming model.

We also remark that the streaming node-churn model,  with different names (e.g. the \textit{sliding-window} model) have been considered for other algorithmic problems: for instance,   \cite{crouch2013dynamic}   considers several graph problem and other problems are studied
in \cite{BorassiELVZ20,BravermanLLM16}.

Some past analytical studies have   focused on distributed algorithms specifically designed to maintain network connectivity under dynamic conditions \cite{DuchonD14,pandurangan2003building}.

A powerful method for maintaining expansion in dynamic networks with churn is based on ID-based random walks. In this approach, each vertex launches $k$ independent random walks carrying its ID. These tokens are mixed throughout the network, and when a new vertex needs to create edges, it connects to the IDs of the tokens it collects. Probabilistic analysis of this method usually shows two key outcomes: the resulting graph has strong expansion, and the random walks become well-distributed quickly  \cite{cooper2007sampling,LS03}.
More in detail, \cite{LS03} provide a distributed algorithm for maintaining
a regular expander in the presence of limited number of insertions and deletions.  The   algorithm is based on a complex procedure that is able to sample uniformly at random  from the space of all possible $2d$-regular graphs formed by $d$ Hamiltonian circuits over the current set of alive nodes. They present
possible distributed implementations of this sample procedure, the best of which, based on random walks,  have $O(\log n)$ overhead and time delay.   Such solutions cannot   manage frequent node churn. 

Further distributed algorithms with different approaches achieving $\bigO(\log n)$ overhead and time delay in the case of slow node churn are proposed in \cite{AS04,RRSST09,MS04,PT14}.
    
In \cite{augustine2015enabling},   an efficient   distributed protocol is introduced that guarantees the maintenance of a bounded degree  topology that, w.h.p., contains an expander subgraph whose set of vertices has size $n-o(n)$. This property is preserved   despite the presence of a large oblivious adversarial churn rate
— up to $\bigO(n/ \polylog(n))$. 
The expander maintenance protocol is efficient  even though  it is rather complex and the local overhead for maintaining the topology is   polylogarithmic in $n$.  A complication of the protocol follows from the fact that, in order to prevent the growth of large clusters of nodes outside the expander subgraph, it  uses   special criteria to ``refresh'' the links of some nodes, even when the latter have not been involved by 
any edge deletion due to the node churn.

Very recently, a new random-walk based protocol for the Poisson node-churn model, is presented in \cite{guptapand2025}. This solution guarantees,  over an   expected churn rate $\Theta(1)$, that the network contains w.h.p.  an expander with a linear number of vertices even in the presence of $o(n/\log n)$ byzantine nodes. This is an important property in some real  network scenarios. To  achieve this property, 
vertices need to   perform  random-walks processes that yield  a communication overhead  $\Theta(n \, \polylog (n))$ per round.
Their  model assumes  the existence of  an \textit{entry manager} that allows every  vertex $v$ to sample a constant number of random neighbors (only) at the time $v$ joins the network. Essentially, the role of the entry manager is equivalent to that of the link manager we adopt in our model but the fact that, in our setting, this service is available for $\Theta(1)$ expected\footnote{And $O(\log n)$ w.h.p. (see \cref{lem:number-times-pending}).}  further calls during the life of $v$. As for this model constraint, we remark  that, in the bitcoin networks \cite{cruciani2023dynamic,nakamoto2008bitcoin}, there are no  kind of  prohibition for using  this service for few more times even after joining the network.  

Finally, recent studies such as \cite{augustine2016distributed} analyzed message flooding in these churn models.  

\section{Conclusion and Open Questions} \label{sec:concl}

The study of dynamic-graphs models capturing key aspects of  real dynamic networks is currently a hot topic in algorithmic research and network science. In what follows,  we discuss  some open questions related to  the  model and the results presented in  this paper.

 We believe it is possible to  extend our analysis on other, more realistic models of node churn, such as the \textit{Poisson} one where nodes enter according to a Poisson clock and have a random age following an  exponential distribution \cite{becchetti2023expansion,pandurangan2003building}. In this setting, the analysis gets  more complicated by two further issues:  the \textit{random} number of nodes each snapshot can have and  the presence of nodes having \textit{random} age, possibly larger than $n$.  However, we think that the   key arguments we used in the analysis of the streaming model can be adapted to take care about such further  issues. Essentially, it could be possible to exploit concentration results on both the number of nodes in a snapshot and on the random life of a node.

 A further interesting scenario is that generated by a different mechanism    to get new link connections. For instance, we can think of a link manager that returns a   non-uniform distribution over the current set of nodes, or that can selects possible links from an  underlying (dynamic) graph  somewhat representing     social relationships among nodes.

 Finally, an important property of distributed protocols is \textit{self-stabilization} \cite{altisen2022introduction,dijkstra1974self}. For short,  it represents the ability of a protocol   to recover its ``good'' behaviour  (guaranteeing some  desired performance and/or property)   from any (worst-case)  configuration   the system can be landed on, due to  some bad event (e.g. a node/link  fault and/or an adversarial setting of some local  variable). The current version of \ALG{} is not fast self-stabilizing under a worst-case scenario where the adversary can corrupt all nodes: essentially, it can  construct a non-expander topology respecting the algorithm rules than can last for a linear number of rounds. However,
      \Cref{lem:decreasing-queue} ensures that the number of pending requests decreases faster: we believe this key-fact can be exploited to design a different, more robust version of \ALG{}\ having fast self-stabilization.

\bibliographystyle{abbrv}
\bibliography{dynamic_graphs_25.bib}

\appendix

\section{Concentration Inequalities} \label{app:prob}
\begin{definition}[Lipschitz property, \cite{dubhashi2009concentration}]
A real-valued function $f(x_1,\dots,x_n)$ satisfies the \emph{Lipschitz property} with constants $d_i$, $i \in [n]$, if 
\[f(\mathbf{x})-f(\mathbf{x'}) \leq d_i\]
whenever $\mathbf{x}$ and $\mathbf{x'}$ differ just in the $i$-th coordinate, $i \in [n]$.
\label{def:lipschitz-property}
\end{definition}

\begin{theorem}[Method of bounded differences, \cite{dubhashi2009concentration}]
If $f$ satisfies the Lipschitz property with constants $d_i$, $i \in [n]$ and $X_1,\dots,X_n$ are independent random variables, then denoting
$f = f(X_1,\dots,X_n)$,
\begin{equation}
    \Prob{f > \Expc{f} + t} \leq {\rm e}^{-\frac{2t^2}{d}} \quad \text{and} \quad \Prob{f < \Expc{f} - t} \leq {\rm e}^{-\frac{2t^2}{d}}
\end{equation}
where $d = \sum_{i=1}^n d_i^2$.
\label{thm:bounded-diff}
\end{theorem}

\begin{theorem}[Chernoff's Inequality]
\label{thm:chernoff}
Let $X=\sum_{i=1}^n X_i$, where all $X_i$ are independently
distributed in $[0,1]$. Let $\mu=\Expc{X}$ and $\mu_- \leq \mu \leq \mu_+$.
Then:
\begin{enumerate}[(a)]
    \item For any $t>0$, it holds
    \[
    \Prob{X>\mu_+ +t}\leq  {\rm e}^{-2t^2/n} \quad \text{and} \quad \Prob{X<\mu_- -t}\leq {\rm e}^{-2t^2/n}.
    \]
    \item For any $\epsilon>0$,
    \[\Prob{X >(1+\epsilon)\mu} \leq {\rm e}^{-\frac{\epsilon^2}{3}\mu} \quad \text{ and }\Prob{X < (1-\epsilon)\mu} \leq {\rm e}^{-\frac{\varepsilon^2}{2}\mu}\]
    \item For $0<\epsilon<1$, it holds
    \[
    \Prob{X>(1+\epsilon)\mu_+}\leq {\rm e}^{-\frac{\epsilon^2}{3}\mu_+} \quad \text{and}\quad \Prob{X<(1-\epsilon)\mu_-}\leq {\rm e}^{-\frac{\epsilon^2}{2}\mu_-}. 
    \]
\end{enumerate}
\end{theorem}

\begin{theorem}[Hoeffding Bound]
\label{thm:hoeffding}
Let $X_1,\dots, X_n$ be independent random variables with such that, for all $i \in [n]$, $\Prob{a_i \leq X_i\leq b_i} = 1$ for constants $a_i$ and $b_i$. Let $X = \sum_{i=1}^n X_i$ and $\mu = \Expc{X}$. Then,
\begin{equation*}
    \Prob{|X - \mu| \geq \varepsilon} \leq 2 {\rm e}^{-\frac{2\epsilon^2}{\sum_{i=1}^n(b_i-a_i)^2}}.
\end{equation*}
\end{theorem}

The following bound gives concentration on the sum of independent identically distributed geometric random variables.

\begin{lemma}
\label{lem:geometric_distribution}
Let $X_1, \dots, X_n$ be a sequence of i.i.d.\ geometric random variables with success probability $p$. Then, we have that
\[\Prob{\sum_{i=1}^n X_i \geq k} = \Prob{\mathrm{Bin}(k,p) \leq n}.\]
\end{lemma}

\begin{proof}
Asking that $\sum_{i=1}^n X_i \geq k$ is like asking that, in $k$ Bernoulli trials, we have less than $n$ successes.
\end{proof}

\end{document}